\documentclass[reprint,aps,prx,superscriptaddress,nobibnotes,floatfix,british,twocolumn]{revtex4-2}

\usepackage[utf8]{inputenc}
\usepackage{amsmath,amssymb,amsthm}
\usepackage{amsfonts}
\usepackage{mathrsfs}
\usepackage{mathtools}
\usepackage{braket}
\usepackage[hidelinks]{hyperref}
\usepackage{cleveref}
\usepackage{enumitem}
\usepackage{tikz}
\usepackage{array}
\usepackage{booktabs}
\usepackage{multirow}
\usepackage{csquotes}
\usepackage{xcolor}
\usepackage{thmtools}
\usepackage{thm-restate}
\usepackage[linesnumbered,ruled,vlined]{algorithm2e}

\newtheorem{theorem}{Theorem}
\newtheorem{lemma}{Lemma}

\newtheorem{definition}{Definition}

\newtheorem{proposition}{Proposition}

\newcommand{\Tr}{\mathrm{Tr}}
\newcommand{\id}{\mathbb{I}}

\newcommand{\calH}{\mathcal{H}}

\newcommand{\calU}{\mathcal{U}}
\newcommand{\calC}{\mathcal{C}}
\newcommand{\calE}{\mathcal{E}}

\newcommand{\calT}{\mathcal{T}}

\newcommand{\TR}[1]{\mathrm{Tr} \left[#1\right]}

\newcommand{\proc}{\mathrm{proc}}
\newcommand{\QT}{\mathrm{QT}}
\newcommand{\Herm}{\mathrm{Herm}}
\newcommand{\RQT}{\mathrm{RQT}}
\newcommand{\ketbra}[2]{\mathinner{|{#1}\rangle\,\langle{#2}|}}

\newcommand{\eff}{\mathrm{eff}}
\newcommand{\inv}{\mathrm{inv}}
\newcommand{\fin}{\mathrm{fin}}

\newcommand{\LGY}{\mathrm{LGYNI}}

\newcommand{\bM}{\mathbb{M}}
\newcommand{\cH}{\mathcal{H}}
\newcommand{\bC}{\mathbb{C}}
\newcommand{\bR}{\mathbb{R}}
\newcommand{\cC}{\mathcal{C}}
\newcommand{\rea}{\mathsf{r}}
\newcommand{\rec}{\mathsf{c}}
\newcommand{\iu}{{i\mkern1mu}}
\newcommand{\enc}{\mathsf{e}}
\newcommand{\dec}{\mathsf{d}}
\newcommand{\tdots}{\mathinner{\ldotp\mkern-3mu\ldotp\mkern-3mu\ldotp}}
\newcommand{\tcdots}{\mathinner{\cdotp\mkern-3mu\cdotp\mkern-3mu\cdotp}}

\definecolor{js}{rgb}{0.578125,0.23828125,0.9609375}
\definecolor{err}{rgb}{1,0,0}

\newcommand{\red}{\color{red}}

\newif\ifshowhierarchyfigure
\showhierarchyfiguretrue      

\begin{document}

	\title{Indefinite Causal Order Reverses the Real-Complex Hierarchy}
	
	\author{Jacopo Surace}
    \email{jacopo.surace@gmail.com}
    \author{Shintaro Minagawa}
    \email{minagawa.shintaro@gmail.com}
    \author{Ravi Kunjwal}
    \email{quaintum.research@gmail.com}
	\affiliation{Aix-Marseille University, CNRS, LIS, Marseille, France}

\begin{abstract}
\textbf{{\red[Note added after submission. After posting the first version of this preprint and corresponding with Ved Kunte and Kuntal Sengupta, we identified an issue with the claimed separation between real quantum theory and ordinary complex quantum theory in the process-matrix framework. This version imposes normalization only for local CPTP maps acting on the parties' process input-output systems, which we call N1. A stronger compositional requirement is to impose normalization also after arbitrary shared ancillary systems are introduced and each party acts jointly on its process system and local share of the ancilla, which we call N2. Under N2, the process matrix used in this version to separate RQT from QT is not valid, and the claimed RQT/QT separation is therefore not established. We are preparing a revised version that clarifies this distinction and revises the affected claims. This temporary update is intended to alert readers to this issue while we work on revising the manuscript.]}}

\vspace{1cm}
Can causal relations be subject to quantum indefiniteness, similar to other physical properties? The process-matrix framework formalises this possibility: valid processes are defined by what local laboratories can implement, without assuming a global causal order. Standardly, the local labs are assumed to implement arbitrary quantum instruments. We ask what happens when symmetries restrict these local operations. Symmetry constraints, such as those arising from missing reference frames, superselection constraints, or the antiunitary symmetry defining real quantum theory, enlarge the admissible process cone. Do these extra processes generate genuinely new correlations? We prove a sharp dichotomy: no for any finite unitary symmetry, yes for real quantum theory. Recent work has shown that, under fixed and definite causal order, complex quantum theory is strictly richer than real quantum theory. Our work shows that this hierarchy is reversed under indefinite causal order: real quantum theory realizes strictly more process correlations than complex quantum theory.
\end{abstract}
	
	\maketitle

\section{Introduction}

Quantum theory overturns our pre-quantum picture of the physical world, one rooted in systems carrying definite physical properties, with no fundamental role for indeterminism, much less for complex numbers~\cite{kochen1967,bell1964}. More than a century after its first appearance, investigations into the foundational aspects of quantum theory continue~\cite{pusey2012,renou2021}. This is due as much to the conceptual challenge of making sense of quantum theory's mathematics, not least its reliance on complex amplitudes~\cite{stueckelberg1960,caves2001,mckague2009,aleksandrova2013,sarkar2025,barriosHita2025,hoffreumonWoods2026,ying2025a}, as it is to the challenge of rigorously reconciling general relativity's dynamical causal structure with quantum theory's probabilistic nature in an account of quantum causality~\cite{hardy2007,oreshkov2012a,brukner2014,hardy2026}. We ask a question at the heart of this two-fold challenge: is there an operational distinction between real and complex quantum theory if causal order is allowed to be indefinite? 

To make this question precise, we adopt the process-matrix framework of Oreshkov, Costa, and Brukner~\cite{oreshkov2012a}, which formalises the possibility of quantum causality with indefinite causal order. The key to the framework is its operational foundation: a process matrix is valid whenever it assigns non-negative, normalized probabilities to every experiment the local agents can actually perform.  This grounding is a strength.  It means that the set of valid processes is not fixed once and for all, but responds to the physical capabilities of the laboratories.   When those capabilities change, so do the admissible processes.

Symmetry restrictions provide a concrete and physically motivated way to change those capabilities. Missing reference frames and superselection rules confine local agents to operations that are invariant or covariant under the relevant symmetry group~\cite{bartlett2007}. This idea has been formalised, in previous work, as \emph{twirled quantum worlds} for unitary symmetry restrictions and as \emph{swirled quantum worlds} for antiunitary symmetry restrictions~\cite{centeno2025,ying2025a}. The canonical antiunitary example is real quantum theory (RQT), obtained by imposing invariance under complex conjugation. In both cases the set of locally available operations is reduced. The consequences for process-matrix correlations, however, are far from obvious, because the restriction acts simultaneously in two competing directions. Fewer available instruments weaken the validity constraints on global processes, admitting representatives that would be forbidden in ordinary quantum process framework. At the same time, fewer instruments mean fewer ways to probe those representatives, potentially rendering the new ones operationally invisible. Which effect dominates?

The answer is sharply different in the unitary and antiunitary cases. Denoting by $\mathcal{C}^{\mathrm{proc,fin}}_{\mathcal{T}}$ the set of probability
distributions achievable in finite-dimensional process experiments in a theory $\mathcal{T}$, we first show that finite unitary symmetry restrictions are operationally neutral. More precisely, for any finite group $G$ acting unitarily on the local systems, if $\mathsf{Tw}_G$ is the corresponding $G$-twirled quantum theory and $\QT$ is ordinary quantum theory, then
\begin{align*}
    \mathcal{C}^{\mathrm{proc,fin}}_{\mathsf{Tw}_G}
    =
    \mathcal{C}^{\mathrm{proc,fin}}_{\mathrm{QT}} .
\end{align*}
Although the symmetry restriction enlarges the cone of valid process representatives, this enlargement does not change the set of observable correlations. The reason is twofold: twirled processes can be symmetrized back to ordinary quantum processes, while finite reference-frame registers allow ordinary quantum process experiments to be simulated within the twirled theory.

The antiunitary case behaves differently.  In real quantum theory, restricting the laboratories to real instruments weakens the constraints that define valid processes in a way that survives at the level of observed probabilities.  Thus the additional real process matrices are not all operationally redundant: some of them generate correlations with real local instruments that no finite-dimensional complex quantum process can reproduce.   We prove
\begin{align*}
    \mathcal{C}^{\mathrm{proc,fin}}_{\mathrm{QT}}
    \subsetneq
    \mathcal{C}^{\mathrm{proc,fin}}_{\mathrm{RQT}} .
\end{align*}
Equivalently, relaxing the causal order allows RQT to generate finite-dimensional process correlations that no complex quantum process can reproduce. The separation is witnessed by an explicit finite-dimensional bipartite RQT process strategy whose value for the Lazy Guess Your Neighbour's Input (LGYNI) causal inequality~\cite{branciard2016simplest} exceeds the dimension-independent upper bound of Liu and Chiribella~\cite{liuChiribella2025} for all finite-dimensional ordinary complex-QT process strategies.

These results place indefinite causal order in a distinctive position within the real--complex comparison.  For twirled theories, our result extends an equivalence already known for fixed causal structures~\cite{ying2025a,bartlett2007}: twirled theories are operationally indistinguishable from ordinary QT at all levels, including genuinely indefinite causal order.  For real quantum theory the picture is richer and more striking. When the causal order between parties is left free but required to be definite, RQT and QT produce identical bipartite process correlations~\cite{mckague2009,aleksandrova2013}.  

Fixing a specific causal structure in a multipartite network breaks this equality: Renou et al.\ showed that, under a fixed bilocal causal structure, $\mathcal{C}^{\mathrm{biloc,fin}}_{\mathrm{RQT}} \subsetneq \mathcal{C}^{\mathrm{biloc,fin}}_{\mathrm{QT}}$~\cite{renou2021,ying2025a}, with QT strictly richer than RQT.  Relaxing the causal order surprisingly reverses this inclusion: we find already in the bipartite scenario that $\mathcal{C}^{\mathrm{proc,fin}}_{\mathrm{QT}} \subsetneq \mathcal{C}^{\mathrm{proc,fin}}_{\mathrm{RQT}}$, with RQT now strictly richer than QT.

\section{Setup}
\label{sec:setup}

We fix notation for process correlations. Each party $X^{(k)}$ has an input Hilbert space $\calH_{X^{(k)}_1}$ and an output Hilbert space $\calH_{X^{(k)}_2}$. For a system $S$, we write $L^S$ for the space of operators on $\calH_S$. Thus $L^{X^{(k)}_1X^{(k)}_2}$ denotes the space of operators on $\calH_{X^{(k)}_1}\otimes\calH_{X^{(k)}_2}$. Local operations are represented by Choi matrices following the convention of Ref.~\cite{oreshkov2012a}: a completely positive map from $X^{(k)}_1$ to $X^{(k)}_2$ corresponds to an operator $M^{X^{(k)}_1X^{(k)}_2}\ge 0$, and it is trace-preserving when $\Tr_{X^{(k)}_2}M^{X^{(k)}_1X^{(k)}_2}= \id^{X^{(k)}_1}$. A local instrument with classical setting $x_k$ and outcome $a_k$ is a family $\{M^{(k)}_{a_k|x_k}\}_{a_k}$ of positive Choi matrices whose sum $M^{(k)}_{x_k}:=\sum_{a_k}M^{(k)}_{a_k|x_k}$ is CPTP.

An $n$-party process matrix is an operator $W\in L^{X^{(1)}_1X^{(1)}_2\cdots X^{(n)}_1X^{(n)}_2}$ that assigns probabilities through the generalized Born rule
\begin{align*}
p(a_1,\tdots,a_n|x_1,\tdots,x_n)
=
\TR{W\!\left(M^{(1)}_{a_1|x_1}\otimes\tcdots\otimes M^{(n)}_{a_n|x_n}\right)}.
\end{align*}
In ordinary complex quantum theory, the operational characterisation of Oreshkov, Costa, and Brukner (OCB) says that requiring non-negative probabilities for all local CP maps, including ancilla-assisted ones, and normalization for all complete local instruments characterises valid process matrices as exactly the operators $W$ satisfying~\cite{oreshkov2012a}
\begin{align}
W&\ge 0,
\label{eq:qt-process-pos}\\
\Tr\!\left[W\left(M_1\otimes\cdots\otimes M_n\right)\right]&=1
\label{eq:qt-process-norm}
\end{align}
for all local CPTP Choi matrices $M_k$. 

\paragraph{Symmetry-restricted laboratories.}
Let $G$ be a finite group acting unitarily on each local system. A local Choi matrix $M^{X_1X_2}$ is $(G,U)$-\emph{covariant} when the corresponding channel satisfies
\begin{align*}
\calE\circ\calU_g^{X_1}=\calU_g^{X_2}\circ\calE
\qquad\forall g\in G,
\end{align*}
where $\calU_g(\cdot)=U_g(\cdot)U_g^\dagger$. A symmetry-restricted laboratory can only implement instruments $\{\mathcal M_a\}_a$ such that each CP map $\mathcal M_a$ is
$(G,U)$-covariant. For several laboratories the same group element $g$ acts on every Choi space at once: if $V_g^{(k)}$ is the induced action on $X^{(k)}_1X^{(k)}_2$, then the global process representative is acted on by $V_g\coloneqq V_g^{(1)}\otimes\cdots\otimes V_g^{(n)}$. The corresponding \emph{$(G,U)$-twirled quantum world} is the subtheory of complex quantum theory with invariant states/effects and covariant transformations~\cite{bartlett2007,centeno2025,ying2025a}. 

A \emph{twirled process matrix} is obtained by applying the OCB operational characterisation to these symmetry-restricted laboratories. Since all non-invariant components are operationally redundant under symmetric tests, we choose globally invariant representatives in
\begin{align*}
L_{\inv} \coloneqq \left\{ W\in L^{X^{(1)}_1X^{(1)}_2\cdots X^{(n)}_1X^{(n)}_2} : V_g W V_g^\dagger = W,\ \forall g\in G \right\}.
\end{align*}
Thus a finite-dimensional $(G,U)$-twirled process representative is an operator $W\in L_{\inv}$ satisfying
\begin{align}
W&\ge 0,
\label{eq:twirled-process-pos}\\
\Tr\!\left[W\left(M_1\otimes\cdots\otimes M_n\right)\right]&=1
\label{eq:twirled-process-norm}
\end{align}
for all local $(G,U)$-covariant CPTP Choi matrices $M_k$. The full proof and the precise induced actions are given in Appendix~\ref{app:characterisation-twirled-W}. Because the normalization in Eq.~\eqref{eq:twirled-process-norm} is imposed on a strictly smaller family of tests than in ordinary QT, the set of admissible twirled process representatives can be strictly larger than the QT process cone. The operational significance of this gap is what the main results address.

\paragraph{Real quantum theory.}
Real quantum theory ($\RQT$) is obtained by restricting complex quantum theory to the sector fixed by time reversal, represented antiunitarily by complex conjugation in fixed local bases~\cite{ying2025a}. Operationally, local states, effects, ancillas, and Choi matrices are real in those bases. For any composite system $S$, define the conjugation-invariant space
\begin{align*}
    L_{\bR}^{S}:=\{O\in L^{S}: O=\overline{O}\}.
\end{align*}
For the full process space we write $\mbox{$L_{\bR}\coloneqq L_{\bR}^{X^{(1)}_1X^{(1)}_2\cdots X^{(n)}_1X^{(n)}_2}$}$.
This space plays for RQT the role played by $L_{\inv}$ in a unitary twirled world.

An \emph{RQT process matrix} is defined by applying the OCB operational characterisation to real laboratories. Since all non-real components are operationally redundant under real tests, we use real representatives $W\in L_{\bR}$. Thus an $n$-party finite-dimensional RQT process representative is an operator $W\in L_{\bR}$ satisfying
\begin{align}
    W&\ge 0,
    \label{eq:rqt-process-pos}\\
    \Tr\!\left[W\left(M_1\otimes\cdots\otimes M_n\right)\right]&=1
    \label{eq:rqt-process-norm-setup}
\end{align}
for all local CPTP Choi matrices $M_i\in L_{\bR}^{X^{(i)}_1X^{(i)}_2}$. Since $W\in L_{\bR}$ and $W\ge0$, this representative is real symmetric. The complete derivation is given in Appendix~\ref{app:characterisation-RQT-W}.

\paragraph{Correlation sets.}
For a theory $\calT$, let $\calC_{\calT}^{\proc}$ denote the set of correlations achievable from process matrices in that theory.  For QT and RQT, we write $\calC_{\calT}^{\proc}(d)$ for the set at fixed local dimensions  $d$ and define the finite-dimensional sets by
\begin{align}
\calC_{\QT}^{\proc,\fin}
:=
\bigcup_{d<\infty}\calC_{\QT}^{\proc}(d),
\qquad
\calC_{\RQT}^{\proc,\fin}
:=
\bigcup_{d<\infty}\calC_{\RQT}^{\proc}(d).
\label{eq:qt-rqt-finite-corr-sets}
\end{align}
For a fixed finite-dimensional twirled world $(G,U)$, the process correlation set at dimension $d$ is $\calC_{(G,U)}^{\proc}(d)$.  Allowing all finite-dimensional systems carrying representations of $G$ gives the full $G$-twirled theory $\mathsf{Tw}_G$, with finite-dimensional process correlation set $\calC_{\mathsf{Tw}_G}^{\proc,\fin}$.

\subsection*{Causal structure of process representatives}
\label{subsec:ocb-forbidden}

Before turning to the main results, we address a concern raised by the enlargement of the admissible representative space in symmetry-restricted theories.

In ordinary QT, the OCB characterisation imposes a stringent structure on valid process matrices: in the bipartite qubit case, only certain Pauli support types are allowed. We call all remaining support types OCB-forbidden. These forbidden terms include components in which a party's output space feeds directly back into its own input without passing through another party's laboratory, and can therefore be interpreted as local causal-loop terms~\cite{oreshkov2012a}.

In symmetry-restricted theories, the normalization constraints are imposed on a strictly smaller set of local tests, and valid process representatives can carry OCB-forbidden support types. We show, however, that in the bipartite qubit case all such components are completely invisible to the restricted local laboratories, because local tomography is incomplete in symmetrized theories.

In ordinary QT, products of local instruments span the full operator space, so every component of a process matrix is, in principle, locally reconstructable. In symmetry-restricted theories, this local tomographic completeness fails~\cite{centeno2025}: the accessible instruments span only a proper subspace of the full operator space. We can therefore decompose any process representative as
\begin{align*}
W = W_{\mathrm{la}} + W_{\mathrm{ga}},
\end{align*}
where $W_{\mathrm{la}}$ is the locally accessible component, reconstructable by local symmetric tomography, and $W_{\mathrm{ga}}$ is the orthogonal component, invisible to all products of locally available tests and only globally accessible. The following then holds:

\begin{proposition}[Causal-loop terms are locally invisible]
\label{prop:ocb-forbidden-invisible}
For bipartite qubit process representatives in finite twirled worlds or in RQT, $W_{\mathrm{la}}$ contains only OCB-allowed Pauli support types. All OCB-forbidden Pauli components, i.e. the causal-loop support types excluded by the ordinary OCB constraints, lie entirely in $W_{\mathrm{ga}}$ and are inaccessible to local symmetric tomography.
\end{proposition}

The proof is given in Appendix~\ref{app:causal-loops}, Theorem~\ref{thm:sym-ocb-forbidden-invisible}. Local agents equipped with symmetric instruments have no means of detecting $W_{\mathrm{ga}}$: it contributes nothing to any observable statistics.

The component $W_{\mathrm{ga}}$ nevertheless plays a structural role. Positivity is imposed on the full representative, $W_{\mathrm{la}}+W_{\mathrm{ga}}$, not on $W_{\mathrm{la}}$ alone. Thus a locally accessible component that would not by itself define a positive ordinary-QT process can become admissible once an appropriate invisible component supplies the missing positivity. Fundamentally, the observable statistics are obtained by projecting the process cone onto the locally accessible subspace, and allowing nonzero $W_{\mathrm{ga}}$ can enlarge this projection.

\section{Finite unitary symmetries}
\label{sec:finite-unitary}

We turn now to the first main result.  For finite unitary symmetry restrictions, the admissible local instruments are covariant ones, and the corresponding process representative cone can be strictly larger than the ordinary QT process cone.  For finite groups and finite-dimensional process experiments, the theorem below establishes that this enlargement leaves the process correlations entirely unchanged.

\begin{theorem}[Finite unitary symmetries are invisible to process correlations]
\label{thm:twirled-operational-equivalence}
Let $G$ be a finite group.  Then
\begin{align}
\calC_{\mathsf{Tw}_G}^{\proc,\fin}
=
\calC_{\QT}^{\proc,\fin}.
\label{eq:twirled-qt-equivalence}
\end{align}
\end{theorem}

The proof is given in Appendix~\ref{app:correlation-sets}, Theorem~\ref{thm:findimequi}.  For the forward inclusion, averaging any twirled process representative over independent local group actions produces an ordinary QT process with identical probabilities on all covariant instruments; this gives $\mbox{$\calC_{(G,U)}^{\proc}(d) \subseteq \calC_{\QT}^{\proc}(d)$}$ at every fixed local input-output dimensions $d$ and representations $(G,U)$.  For the reverse inclusion, finite-dimensional reference-frame registers carrying the regular representation of $G$ allow any QT process experiment to be encoded faithfully inside $\mathsf{Tw}_G$, establishing equality at the level of $\calC_{\mathsf{Tw}_G}^{\proc,\fin}$.

The extra representatives admitted by a finite unitary symmetry restriction therefore carry no additional correlation content.  No process experiment, however cleverly designed within the twirled theory, can produce correlations beyond those available in ordinary QT.

\section{Real quantum theory}
\label{sec:rqt-corr}

We turn now to the second main result.  The antiunitary restriction given by complex conjugation yields real quantum theory, in which states, effects, Choi matrices, and process representatives are all real in a fixed basis~\cite{ying2025a}.  The reference-frame simulation argument of the previous section relies essentially on finite unitary group actions and does not extend to this setting.  In the bipartite finite-dimensional setting considered below, the following theorem shows that the enlarged RQT process cone has operationally visible consequences in the process-matrix framework.

\begin{theorem}[Real process correlations strictly contain complex ones]
\label{thm:QT-strict-RQT}
For bipartite finite-dimensional process-matrix correlations,
\begin{align}
\calC_{\QT}^{\proc,\fin}
\subsetneq
\calC_{\RQT}^{\proc,\fin}.
\label{eq:QT-strict-RQT}
\end{align}
\end{theorem}

The inclusion follows from an extension of the standard realification of complex systems, instruments, and processes~\cite{stueckelberg1960,mckague2009,ying2025a}.
Strictness is certified in the End Matter by an explicit finite-dimensional bipartite RQT correlation whose LGYNI value exceeds the dimension-independent complex-QT process bound. The inclusion is proved in Appendix~\ref{app:correlation-sets}, Theorem~\ref{thm:QT-strict-RQT-supp}; the explicit separating distribution is given in the End Matter and the numerical realization is given in Appendix~\ref{app:correlation-sets}.

\ifshowhierarchyfigure
\begin{figure}[t]
    \centering
\begin{tikzpicture}[scale=0.90, every node/.style={transform shape}]
    \definecolor{rqtpurple}{RGB}{190,145,245}
    \definecolor{rqtpurpledark}{RGB}{95,45,150}
    \definecolor{qtgreen}{RGB}{120,245,185}
    \definecolor{qtgreendark}{RGB}{20,120,75}
    \definecolor{pointorange}{RGB}{230,110,40}

    \begin{scope}[shift={(0,2.35)}]
        \node[
            font=\bfseries,
            anchor=east
        ] at (-2.15,0) {Fixed causal structure};

        \filldraw[
            fill=rqtpurple!35,
            draw=rqtpurpledark,
            thick
        ] (0,0) ellipse (1.75cm and 0.90cm);

        \node[
            text=rqtpurpledark,
            font=\large\bfseries
        ] at (0,0.45) {$\mathcal{C}_{\mathrm{QT}}$};

        \filldraw[
            fill=qtgreen!55,
            draw=qtgreendark,
            thick
        ] (0,-0.28) ellipse (1.05cm and 0.43cm);

        \node[
            text=qtgreendark,
            font=\large\bfseries
        ] at (0,-0.28) {$\mathcal{C}_{\mathrm{RQT}}$};
    \end{scope}

    \draw[->, thick, >=stealth] (0,1.25) -- (0,0.95);

    \begin{scope}[shift={(0,0)}]
        \node[
            font=\bfseries,
            anchor=east
        ] at (-2.15,0) {Definite causal order};

        \filldraw[
            fill=rqtpurple!35,
            draw=rqtpurpledark,
            thick
        ] (0,0) ellipse (1.75cm and 0.90cm);

        \filldraw[
            fill=qtgreen!45,
            draw=qtgreendark,
            thick,
            opacity=0.78
        ] (0,0) ellipse (1.75cm and 0.90cm);

        \node[
            text=black,
            font=\large\bfseries,
            align=center
        ] at (0.0,0.0)
        {$\mathcal{C}_{\mathrm{QT}}=\mathcal{C}_{\mathrm{RQT}}$};


    \end{scope}

    \draw[->, thick, >=stealth] (0,-1.10) -- (0,-1.40);

    \begin{scope}[shift={(0,-2.35)}]
        \node[
            font=\bfseries,
            anchor=east
        ] at (-2.15,0) {Indefinite causal order};

        \filldraw[
            fill=qtgreen!45,
            draw=qtgreendark,
            thick
        ] (0,0) ellipse (1.75cm and 0.90cm);

        \node[
            text=qtgreendark,
            font=\large\bfseries
        ] at (0,0.50) {$\mathcal{C}_{\mathrm{RQT}}$};

        \filldraw[
            fill=rqtpurple!35,
            draw=rqtpurpledark,
            thick
        ] (0,-0.28) ellipse (1.05cm and 0.43cm);

        \node[
            text=rqtpurpledark,
            font=\large\bfseries
        ] at (0,-0.30) {$\mathcal{C}_{\mathrm{QT}}$};

        \filldraw[
            fill=pointorange,
            draw=black,
            line width=0.5pt
        ] (1.25,0.08) circle (0.07cm);

        \node[
            font=\bfseries,
            text=black
        ] at (2.2,0.08) {$p_\star$};

        \draw[
            ->,
            thick,
            >=stealth
        ] (2.0,0.08) -- (1.34,0.08);
    \end{scope}

\end{tikzpicture}
\caption{
Real--complex hierarchy across causal assumptions. For fixed definite multipartite causal structures, QT is strictly richer than RQT
\cite{renou2021}. For bipartite definite-order correlations, RQT and QT coincide~\cite{mckague2009,aleksandrova2013}.
For process correlations with indefinite causal order, the hierarchy reverses: RQT is strictly richer than QT, as shown in this work. The point $p_\star$
denotes the separating correlation constructed here. 
}
\label{fig:real-complex-hierarchy-reversal}
\end{figure}
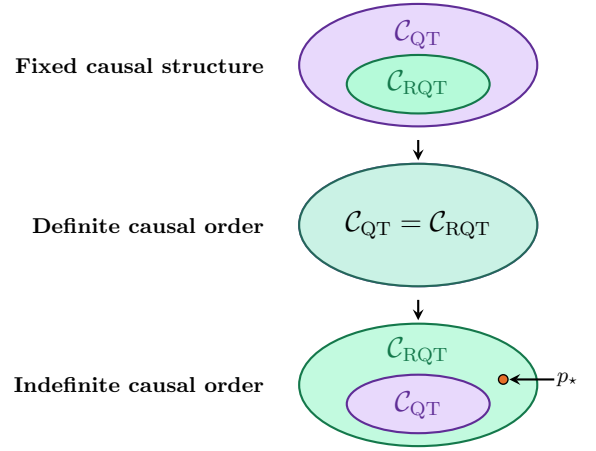
\fi

The inclusion in Eq.~\eqref{eq:QT-strict-RQT} contrasts with fixed definite causal order network separations.  In multipartite networks with a fixed definite causal order, complex quantum theory can be strictly richer than real quantum theory~\cite{renou2021,ying2025a}.  Here, removing the fixed background causal order changes the operational hierarchy: real quantum theory realizes process correlations outside finite-dimensional complex quantum theory.

\section{Conclusion}
\label{sec:conclusion}

We have characterised finite-dimensional process-matrix correlations under finite symmetry restrictions, organizing them according to their unitary or antiunitary implementation.

For finite unitary symmetry groups, such restrictions, however severe at the level of local operations, are operationally neutral:
\begin{align*}
\mathcal{C}^{\mathrm{proc,fin}}_{\mathsf{Tw}_G}
=
\mathcal{C}^{\mathrm{proc,fin}}_{\mathrm{QT}} .
\end{align*}
By contrast, in the bipartite finite-dimensional process setting, the antiunitary restriction defining real quantum theory has operational consequences:
\begin{align*}
\mathcal C^{\mathrm{proc,fin}}_{\QT}\subsetneq \mathcal C^{\mathrm{proc,fin}}_{\RQT}.
\end{align*}
Thus the operational effect of a symmetry restriction depends on both the restriction of the local laboratories and the unitary or antiunitary character of the symmetry.

The striking feature of the real case is the direction of the separation. For  fixed definite causal order, real quantum theory is strictly weaker than complex quantum theory: some complex-QT network correlations cannot be reproduced in RQT~\cite{renou2021,ying2025a}. In the bipartite process-matrix setting, the inclusion is reversed. Once the causal order is not fixed, RQT realizes finite-dimensional process correlations that no finite-dimensional complex quantum process can reproduce. Thus indefinite causal order changes not only the size of the correlation set, but also the relative ordering of the real and complex theories.

The common starting point is the failure of local tomography. In both twirled theories and RQT, the locally available instruments span only a proper subspace of the full operator space~\cite{centeno2025}. A process representative may therefore contain components that are invisible to all local tests, including OCB-forbidden components associated with local causal-loop terms. These components have no observable signature of their own, but they can supply positivity support for the locally accessible part of the process.

What differs is whether this extra freedom survives at the level of observable correlations. For finite unitary symmetries, it is neutralized by finite reference-frame simulations, and the process correlations remain exactly those of ordinary QT. For real quantum theory, it is not neutralized: the hidden positivity support allows locally accessible process components that lead to correlations beyond finite-dimensional complex-QT processes. 
In fixed definite causal order networks, complex QT is strictly richer than RQT~\cite{renou2021}; in the bipartite process-matrix setting, the inclusion is reversed.

We note that our conclusions are  tied to the standard process-matrix notion of locality: laboratories are local because their input and output spaces are tensor factors. Alternative formulations of RQT that modify composition or source independence~\cite{barriosHita2025,hoffreumonWoods2026} may therefore offer a way around the separation, but only by changing the notion of locality, and hence the definition of a process matrix itself. This makes indefinite causal order a sharp testing ground for such formulations: beyond recovering complex-QT behavior in definite-order scenarios, they must say what a process is when the causal order is not fixed, and what new structures this reveals.

Several directions remain open. Our finite reference-frame argument does not cover infinite or continuous symmetry groups, such as spatial rotations, and it is unknown whether the equivalence with QT persists there. Beyond symmetry restrictions, one could ask whether the dependence on causal assumptions found here for real versus complex quantum theory has an analogue for quaternionic quantum theory~\cite{barnum2015}, or more generally for non-quantum local theories. Recent work on indefinite causal order in boxworld theories provides a complementary direction along these lines~\cite{bavaresco2024}. A broader goal is to identify which features of the local theory determine whether the enlarged process cone is operationally neutralized, as for finite unitary symmetries, or becomes visible in correlations, as for real quantum theory.

Another natural direction concerns antinomicity~\cite{kunjwalOreshkov2023,kunjwalOreshkov2024}. In bipartite scenarios, antinomicity coincides with causal-inequality violation, so the LGYNI separation already shows that real process matrices can be strictly more antinomic than complex ones. In tripartite scenarios, where antinomicity is strictly stronger than noncausality, one can ask whether a real process can violate a properly nomic inequality~\cite{wechsAhmedLiuKunjwalForthcoming}.

\section*{Data Availability}
The code and data are available at~\cite{code}. The version corresponding to the preprint is ``v0.1-preprint.''

\section*{Acknowledgments}
This work received support from the French government under the France 2030 investment plan, as part of the Initiative d’Excellence d’Aix-Marseille Université-A*MIDEX, AMX-22-CEI-01

	\bibliography{RAVI.bib}

    \onecolumngrid
    \section{END MATTER}

    \twocolumngrid
    \textit{End Matter: Strict inclusion in Theorem~\ref{thm:QT-strict-RQT}}---.
    Suppose that Alice chooses a measurement setting $x$ and applies a CP-instrument with the input system $A_1$ and output system $A_2$ whose Choi operators are $\bM^A:=\{M^{A_1A_2}_{a|x}\}_{a,x}$.
    Bob chooses a setting $y$ and applies an instrument with the input system $B_1$ and output system $B_2$ whose Choi operators are $\bM^B:=\{M^{B_1B_2}_{b|y}\}_{b,y}$.

    The bipartite conditional probability distribution obtained in the process matrix form is
    \begin{equation*}
        p(a,b|x,y)=\Tr\left[W\left(M^{A_1A_2}_{a|x}\otimes M^{B_1B_2}_{b|y}\right)\right].
    \end{equation*}
    It is enough to find one correlation function $\sum_{a,b,x,y}\alpha_{abxy}p(a,b|x,y)$ whose achievable value in RQT goes beyond an upper bound of QT.

    The one example is the Lazy Guess Your Neighbour's Input (LGYNI) correlation function~\cite{branciard2016simplest}:
    \begin{equation*}
        \begin{split}
            I_{\LGY}(p)&=\frac{1}{4}[1+p(0,0|0,1)+p(1,0|0,1)+p(0,0|1,0)\\
            &\quad+p(0,1|1,0)+p(1,1|1,1)],
        \end{split}
    \end{equation*}
    which is bounded by $3/4$ for causal correlations.
    This is exactly the LGYNI functional optimized by Liu and Chiribella~\cite{liuChiribella2025}. 
    Their success probabilities are identified with $P_{\mathrm{succ}}^{01}=p(0,0|0,1)+p(1,0|0,1)$, $P_{\mathrm{succ}}^{10}=p(0,0|1,0)+p(0,1|1,0)$, and $P_{\mathrm{succ}}^{11}=p(1,1|1,1)$, so that their objective is precisely the functional \(I_{\LGY}\) above.
    Ref.~\cite{liuChiribella2025} optimizes this functional over ordinary complex quantum process correlations with arbitrary finite-dimensional local Hilbert spaces, arbitrary complex quantum instruments, and valid complex process matrices in the standard OCB sense. 
    This is precisely the correlation set denoted here by \(\mathcal C^{\mathrm{proc,fin}}_{\mathrm{QT}}\). 
    Therefore their dimension-independent bound applies directly to every \(p\in\mathcal C^{\mathrm{proc,fin}}_{\mathrm{QT}}\):
    \begin{equation*}
        I_{\LGY}(p)\le I_{\LGY}^{\mathrm{ICO}}\simeq 0.8194 .
    \end{equation*}
    The comparison is made at the level of the operational probabilities \(p(a,b|x,y)\), so possible differences in Choi-matrix conventions amount only to fixed relabellings or transpositions of process and instrument representatives and do not change the value of \(I_{\LGY}\).
    

    \begin{table}[tbp]
    \centering
    \begin{tabular}{ccccc}
    \toprule
    $(x,y)$
    & $p(0,0|x,y)$
    & $p(0,1|x,y)$
    & $p(1,0|x,y)$
    & $p(1,1|x,y)$ \\
    \midrule
    $(0,0)$ & $0.250000134$ & $0.249999867$ & $0.250000133$ & $0.249999866$ \\
    $(0,1)$ & $0.443682389$ & $0.056317613$ & $0.443682386$ & $0.056317612$ \\
    $(1,0)$ & $0.443682625$ & $0.443682151$ & $0.056317642$ & $0.056317582$ \\
    $(1,1)$ & $0.163750682$ & $0.085116869$ & $0.085037546$ & $0.666094904$ \\
    \bottomrule
    \end{tabular}
    \caption{
    Full distribution $p_\star$ achieving an LGYNI value of
    $0.8602061136164$. All entries are rounded to nine decimal places.
    }
    \label{tab:rqt-lgyni-distribution}
\end{table}
    In RQT, the explicit finite-dimensional process and real instruments reported in Appendix~\ref{app:correlation-sets} generate the probability distribution in Table~\ref{tab:rqt-lgyni-distribution}, with
    \begin{equation*}
        I_{\LGY}(p_\star)=0.8602061136164.
    \end{equation*}
    Since the bound of Ref.~\cite{liuChiribella2025} holds after optimizing over all finite local dimensions, all ordinary complex instruments, and all valid ordinary complex process matrices, no finite-dimensional complex-QT process strategy can reproduce a distribution with this LGYNI value.
    Indeed, this value exceeds the dimension-independent complex-QT process bound by approximately \(4.08\times 10^{-2}\), so
    \begin{equation*}
        p_\star\notin \mathcal C^{\mathrm{proc,fin}}_{\mathrm{QT}} .
    \end{equation*}
    On the other hand, the explicit realization in Appendix~\ref{app:correlation-sets} gives \(p_\star\in \mathcal C^{\mathrm{proc,fin}}_{\mathrm{RQT}}\). 
    Together with the inclusion \(\mathcal C^{\mathrm{proc,fin}}_{\mathrm{QT}}\subseteq \mathcal C^{\mathrm{proc,fin}}_{\mathrm{RQT}}\) proved in Appendix~\ref{app:correlation-sets}, this establishes the strict inclusion in Theorem~\ref{thm:QT-strict-RQT}.


    \begin{table}[t]
    \centering
    \begin{tabular}{cc}
    \toprule
    Pauli string & coefficient \\
    \midrule
    $X^{A_1}Y^{A_2}X^{B_1}Y^{B_2}$ & $-0.037494312$ \\
    $Y^{A_1}Z^{A_2}X^{B_1}Y^{B_2}$ & $+0.035808848$ \\
    $X^{A_1}Y^{A_2}Y^{B_1}Z^{B_2}$ & $-0.034888881$ \\
    $Y^{A_1}Z^{A_2}Y^{B_1}Z^{B_2}$ & $+0.033317409$ \\
    $X^{A_1}Y^{A_2}Y^{B_1}X^{B_2}$ & $-0.013800738$ \\
    $Y^{A_1}Z^{A_2}Y^{B_1}X^{B_2}$ & $+0.013180541$ \\
    $Y^{A_1}X^{A_2}X^{B_1}Y^{B_2}$ & $+0.011117532$ \\
    $Y^{A_1}X^{A_2}Y^{B_1}Z^{B_2}$ & $+0.010343143$ \\
    $Y^{A_1}X^{A_2}Y^{B_1}X^{B_2}$ & $+0.004091046$ \\
    $X^{A_1}Y^{A_2}Z^{B_1}Y^{B_2}$ & $-0.000963130$ \\
    $Y^{A_1}Z^{A_2}Z^{B_1}Y^{B_2}$ & $+0.000921186$ \\
    $Y^{A_1}X^{A_2}Z^{B_1}Y^{B_2}$ & $+0.000285984$ \\
    \bottomrule
    \end{tabular}
    \caption{
    Largest (top 12) Pauli coefficients of the obtained RQT process matrix outside
    the ordinary complex-QT bipartite process-matrix subspace. The entries
    are ordered by the absolute values of the coefficients. The Pauli-string
    order is $A_1,A_2,B_1,B_2$, and the coefficients are rounded to nine
    decimal places.
    }
    \label{tab:forbidden-pauli}
    \end{table}

    As a diagnostic of the RQT process matrix obtained by the see-saw search, we decomposed the obtained process matrix $W$ in the four-qubit Pauli basis. 
    The Pauli string
    order is $A_1,A_2,B_1,B_2$, where $A_1,B_1$ are input systems and
    $A_2,B_2$ are output systems. 
    Explicitly, we write
    \begin{equation*}
    W=\sum_{\mu,\nu,\kappa,\lambda\in\{I,X,Y,Z\}}w_{\mu\nu\kappa\lambda}\,\sigma_\mu^{A_1}\otimes\sigma_\nu^{A_2}\otimes\sigma_\kappa^{B_1}\otimes\sigma_\lambda^{B_2},
    \end{equation*}
    with
    \begin{equation*}
    w_{\mu\nu\kappa\lambda}=\frac{1}{16}\Tr\!\left[\left(\sigma_\mu^{A_1}\otimes\sigma_\nu^{A_2}\otimes\sigma_\kappa^{B_1}\otimes\sigma_\lambda^{B_2}\right)W\right].
    \end{equation*}
    According to Ref.~\cite{oreshkov2012a}, for ordinary complex quantum theory, the bipartite process-matrix
    constraints allow only Pauli strings whose support is contained in
    \begin{equation*}
    \begin{split}
    &\mathcal S_{\mathrm{OCB}}=\\
    &\{\emptyset,\,A_1,\,B_1,\,A_1B_1,\,A_2B_1,\,A_1B_2,\,A_1A_2B_1,\,A_1B_1B_2\}.
    \end{split}
    \end{equation*}
    We call a Pauli string OCB-allowed if its support belongs to
    $\mathcal S_{\mathrm{OCB}}$, and OCB-forbidden otherwise.
    


The obtained RQT process matrix contains non-negligible OCB-forbidden coefficients, the largest of which are listed in Table~\ref{tab:forbidden-pauli}. 
All these terms have full support on $A_1A_2B_1B_2$, corresponding to global-loop terms in the sense of Ref.~\cite{oreshkov2012a}. 
Their Pauli structure makes them compatible with RQT while keeping them outside the locally accessible sector: each listed string contains one $Y$ factor in Alice's laboratory and one in Bob's laboratory. 
For example, the term $Y^{A_1}Z^{A_2}X^{B_1}Y^{B_2}$ has local factors $Y^{A_1}Z^{A_2}$ and $X^{B_1}Y^{B_2}$, so
\begin{align*}
\Tr\!\left[M^{A_1A_2}(Y^{A_1}\otimes Z^{A_2})\right]=0
\end{align*}
for every real Choi operator $M^{A_1A_2}$, with the analogous statement holding for Bob. 
At the same time, the full string contains an even number of $Y$ factors and is therefore globally real, so it is allowed in an RQT process representative while being inaccessible to local real tomography.

These coefficients expose the mechanism behind the separation. 
With the notation of the main text, the probabilities in Table~\ref{tab:rqt-lgyni-distribution} are determined by $W_{\mathrm{la}}$, while the OCB-forbidden components lie in $W_{\mathrm{ga}}$. 
The latter provide positivity support for the full process operator $W$, making admissible a locally accessible component $W_{\mathrm{la}}$ that realizes correlations beyond the complex-QT process bound. 
Thus the OCB-forbidden terms enlarge the projection of the admissible RQT process cone onto the locally accessible real subspace, while the LGYNI violation is read out through $W_{\mathrm{la}}$.

\clearpage
\newpage

\onecolumngrid
\appendix 

\section{Notation and conventions}
	\label{app:notation}
	
    We work throughout with finite-dimensional Hilbert spaces. For a system $S$, we write $L^S$ for the vector space of linear operators on $\mathcal H_S$. For a party $X$, we write $X=(X_1,X_2)$ for its input and output systems, so $L^{X_1X_2}$ denotes the vector space of linear operators on $\mathcal H_{X_1}\otimes\mathcal H_{X_2}$. More generally, $L^{X^{(1)}_1X^{(1)}_2\cdots X^{(n)}_1X^{(n)}_2}$ denotes the corresponding vector space of linear operators on the tensor-product Hilbert space. The identity is denoted by $\id$, the full trace by $\Tr$, and partial traces by $\Tr_Y$.
	
	We fix once and for all the computational basis used to define transpose and complex conjugation. Accordingly, $\overline{X}$, $X^T$, and $X^\dagger$ denote, respectively, entrywise complex conjugation, transpose, and adjoint in that basis. In particular, in RQT an operator is real iff $X=\overline{X}$.
	
	We use the Choi--Jamio\l{}kowski convention of Ref.~\cite{oreshkov2012a}. For a linear map $\mathcal{M}:L(\mathcal{H}_{X_1})\to L(\mathcal{H}_{X_2})$, its Choi operator is
	\begin{align*}
		M^{X_1X_2}:=\big[(\mathcal{I}\otimes \mathcal{M})(|\phi^+\rangle\langle\phi^+|)\big]^T,\qquad |\phi^+\rangle=\sum_j |jj\rangle.
	\end{align*}
	Equivalently,
	\begin{align*}
		M^{X_1X_2}=\sum_{i,j}|j\rangle\langle i|^{X_1}\otimes {\mathcal{M}(|i\rangle\langle j|)^T}^{X_2}.
	\end{align*}
	In this convention,
	\begin{align*}
		\mathcal{M}(\sigma)=\Tr_{X_1}\!\left[M^{X_1X_2}(\sigma\otimes \id)\right]^T.
	\end{align*}
	Also, $\mathcal{M}$ is completely positive iff $M^{X_1X_2}\ge 0$, and it is trace preserving iff $\Tr_{X_2}M^{X_1X_2}=\id^{X_1}$.
	
	A physical symmetry $g\in G$ acts on operators as $\mathcal{U}_g^X(O):=U_g^X O\,U_g^{X\dagger}$. In our Choi convention, the induced action on Choi operators is $M\mapsto V_g^X M (V_g^X)^\dagger$, where
	\begin{align*}
		V_g^X:=U_g^{X_1}\otimes \overline{U}_g^{X_2}.
	\end{align*}
	Thus, if $\mathcal{M}_g=\mathcal{U}_g^{X_2}\circ \mathcal{M}\circ \mathcal{U}_{g^{-1}}^{X_1}$, then its Choi operator is
	\begin{align*}
		M_g^{X_1X_2}=V_g^X M^{X_1X_2}(V_g^X)^\dagger.
	\end{align*}
	Accordingly, the local twirling map on Choi operators is
	\begin{align*}
		\mathcal{T}^X(M):=\int_G dg\,V_g^X M (V_g^X)^\dagger,
	\end{align*}
	with the integral replaced by the normalized sum when $G$ is finite.
	
	For any composite system $S$ carrying a unitary representation $V_g^S$, we write
\begin{align*}
    L_{\mathrm{inv}}^{S}
    :=
    \{M\in L^{S}:V_g^S M (V_g^S)^\dagger=M\ \forall g\in G\}
\end{align*}
for the invariant subspace. For any composite system $S$ equipped with the fixed product basis used for complex conjugation, we write
\begin{align*}
    L_{\mathbb R}^{S}
    :=
    \{M\in L^{S}:M=\overline M\}
\end{align*}
for the real subspace. Both $L_{\mathrm{inv}}^{S}$ and $L_{\mathbb R}^{S}$ are subspaces of the complex-linear operator space $L^S$. The self-adjoint parts used for Choi operators, process matrices, and tomography are
\begin{align*}
    \Herm^{S}
    &:=
    \{M\in L^{S}:M=M^\dagger\},\\
    \Herm_{\mathrm{inv}}^{S}
    &:=
    L_{\mathrm{inv}}^{S}\cap \Herm^{S},\\
    \Herm_{\mathbb R}^{S}
    &:=
    L_{\mathbb R}^{S}\cap \Herm^{S}.
\end{align*}
Thus, for a local input-output pair $X=(X_1,X_2)$, the locally accessible self-adjoint operator spaces are $\Herm_{\mathrm{inv}}^{X_1X_2}$ in a twirled world and $\Herm_{\mathbb R}^{X_1X_2}$ in RQT.
	
	Throughout the paper, a \emph{twirled process matrix} means the chosen globally invariant representative, while an RQT process matrix means the chosen real-symmetric representative of the corresponding operational equivalence class. These are global representative conventions: local symmetric tomography may access only a smaller subspace, namely the tensor product of the local symmetric self-adjoint spaces. Thus the characterisation theorems concern these fixed representative conventions, not arbitrary raw operators prior to quotienting by operational indistinguishability.
	
	In the bipartite qubit case we use the Pauli basis $\{\sigma_0,\sigma_1,\sigma_2,\sigma_3\}=\{\id,X,Y,Z\}$. Greek indices $\mu,\nu,\lambda,\gamma$ range over $\{0,1,2,3\}$, while Latin indices $i,j,k,l$ range over $\{1,2,3\}$. A support label such as $A_1$, $A_2B_1$, or $A_1B_1B_2$ means that the corresponding Pauli string is nontrivial  on those tensor factors and equal to the identity on the others. The terminology \emph{OCB-allowed} and \emph{OCB-forbidden} is used only in this bipartite qubit sense.
	
    For a theory $\mathcal{T}$, we write $\mathcal{C}_{\mathcal{T}}^{\mathcal{D}}$ for the set of correlations compatible with a fixed DAG $\mathcal{D}$, $\mathcal{C}_{\mathcal{T}}^{\mathrm{DAG}}$ for the union over all DAGs, and $\mathcal{C}_{\mathcal{T}}^{\proc}$ for the set obtainable from arbitrary process matrices. When needed, $W_{\mathrm{la}}^{\mathrm{sym}}$ denotes the orthogonal projection of a process matrix $W$ onto the locally accessible self-adjoint symmetric subspace, namely $\Herm_{\mathrm{inv}}^{A_1A_2}\otimes \Herm_{\mathrm{inv}}^{B_1B_2}$ in a twirled world and $\Herm_{\mathbb{R}}^{A_1A_2}\otimes \Herm_{\mathbb{R}}^{B_1B_2}$ in RQT.

	\subsection{Process matrices}
	\label{sec:pm-review}
	
	We recall the process-matrix framework of Ref.~\cite{oreshkov2012a}. Consider two parties, Alice and Bob, performing local operations in their laboratories, without assuming any underlying global causal order. Let $\mathcal H_{A_1}$ and $\mathcal H_{A_2}$ denote Alice's input and output Hilbert spaces, and similarly for Bob. We denote by $L^{A_1}$, $L^{A_2}$, and $L^{A_1A_2}$ the corresponding vector spaces of linear operators; similarly for Bob. 
	
	Each laboratory is described by quantum instruments. For Alice, these are collections of CP maps $\{\mathcal{M}_j^A\}$ from $L^{A_1}$ to $L^{A_2}$ such that $\sum_j\mathcal{M}_j^A$ is CPTP. By the Choi--Jamio\l{}kowski isomorphism, each CP map is represented by a positive operator $M_j^{A_1A_2}\in L^{A_1A_2}$, and the CPTP condition reads
	\begin{align}
		M^{A_1A_2}\ge 0,\qquad \Tr_{A_2}M^{A_1A_2}=\id^{A_1}. \label{eq:cptp-Choi}
	\end{align}
	The same holds for Bob. We use throughout the Choi--Jamio\l{}kowski convention of Ref.~\cite{oreshkov2012a}; see Section~\ref{app:notation}.
	
	A process matrix is an operator $W\in L^{A_1A_2B_1B_2}$ such that the joint outcome probabilities are given by the generalized Born rule
	\begin{align}
		P(i,j)=\Tr\!\bigl[W(M_i^{A_1A_2}\otimes M_j^{B_1B_2})\bigr], \label{eq:gen-Born}
	\end{align}
	and satisfy the following operational requirements:
	
	\begin{enumerate}
		\item[\textbf{(O1)}] \emph{Positivity:} $P(i,j)\ge 0$ for all local CP maps, including ancilla-assisted ones;
		\item[\textbf{(O2)}] \emph{Normalization:} $\sum_{i,j}P(i,j)=1$ for all pairs of complete local instruments.
	\end{enumerate}
	
	In QT, these conditions are equivalent to the standard characterisation~\cite{oreshkov2012a}.
	
	\begin{proposition}[Process matrices~\cite{oreshkov2012a}]
		\label{prop:W-charac}
		A matrix $W\in L^{X^{(1)}_1X^{(1)}_2\cdots X^{(n)}_1X^{(n)}_2}$ satisfies $(\mathrm{O}1)$ and $(\mathrm{O}2)$ iff
		\begin{align}
			W&\ge 0, \label{eq:W-pos}\\
			\Tr\!\bigl[W(M_1\otimes\cdots\otimes M_n)\bigr]&=1 \label{eq:W-norm}
		\end{align}
		for all local CPTP Choi matrices $M_k$.
	\end{proposition}
	
	We will later refer to the standard support-pattern description in the bipartite qubit case. Expanding
	\begin{align*}
		W=\sum_{\mu,\nu,\lambda,\gamma=0}^3 w_{\mu\nu\lambda\gamma}\,\sigma_\mu^{A_1}\otimes \sigma_\nu^{A_2}\otimes \sigma_\lambda^{B_1}\otimes \sigma_\gamma^{B_2},
	\end{align*}
	with $\{\sigma_0,\sigma_1,\sigma_2,\sigma_3\}=\{\id,X,Y,Z\}$, the normalization condition~\eqref{eq:W-norm} allows  the support types $1$, $A_1$, $B_1$, $A_1B_1$, $A_2B_1$, $A_1B_2$, $A_1A_2B_1$, and $A_1B_1B_2$~\cite{oreshkov2012a}. We call these \emph{OCB-allowed}. All remaining bipartite qubit support types are excluded in ordinary QT; we call them \emph{OCB-forbidden}. Each forbidden pattern is associated with causal cyclicity or postselection, so their automatic exclusion by the conditions of Proposition~\ref{prop:W-charac} can be seen as a consistency property of the framework.
	
	\subsection{Twirled and swirled quantum worlds}
	\label{sec:twirl-swirl-review}
	
	The process-matrix framework assumes that local laboratories have access to the full set of quantum instruments. If the parties lack a shared reference frame, the operationally accessible states, effects, and transformations are restricted to those symmetric under the relevant group action, equivalently to a superselection rule~\cite{bartlett2007}. Centeno \emph{et al.}~\cite{centeno2025} formalised this as a \emph{twirled quantum world}, and Ying \emph{et al.}~\cite{ying2025a} extended it to antiunitary symmetries under the name of \emph{swirled quantum world}. Because Proposition~\ref{prop:W-charac} quantifies over all local instruments, restricting to symmetric ones can enlarge the set of admissible process matrices.
	
	Let $G$ be a compact group acting on a system $S_i$ through a unitary representation $\{U_g^{S_i}\}_{g\in G}$, and define $\mathcal{U}_g^{S_i}(\cdot):=U_g^{S_i}(\cdot)U_g^{S_i\dagger}$. For a composite system $S=S_1\cdots S_n$, the action is $\mathcal{U}_g^S=\mathcal{U}_g^{S_1}\otimes\cdots\otimes \mathcal{U}_g^{S_n}$. An operator $O^S$ is $(G,U)$-invariant if $\mathcal{U}_g^S(O^S)=O^S$ for all $g$, and a channel $\mathcal{E}^{S'|S}$ is $(G,U)$-covariant if $\mathcal{E}^{S'|S}\circ \mathcal{U}_g^S=\mathcal{U}_g^{S'}\circ \mathcal{E}^{S'|S}$ for all $g$. The $(G,U)$-twirled quantum world is the subtheory consisting of invariant states and effects and covariant operations~\cite{centeno2025}.
	
	For an antiunitary symmetry, the analogous construction defines a swirled quantum world~\cite{ying2025a}. The relevant case here is time reversal, represented by $\{I,C\}$ with $C$ complex conjugation in a fixed basis. The corresponding swirled theory is real quantum theory (RQT): the subtheory in which all states, effects, and Choi matrices are real in that basis.

	\subsection{Twirled process matrices}
	\label{sec:twirled-pm}
	
	In a $(G,U)$-twirled quantum world, local instruments, ancillas, and operations are restricted to $(G,U)$-invariant or $(G,U)$-covariant ones. We define twirled process matrices by imposing the usual operational conditions $(\mathrm{O}1)$ and $(\mathrm{O}2)$ only on this restricted class.
	
	Since we are interested in the operationally accessible part of the process, we work throughout with the globally invariant representative, namely with operators in the invariant subspace $L_{\mathrm{inv}}^{X^{(1)}_1X^{(1)}_2\cdots X^{(n)}_1X^{(n)}_2}$. Equivalently, throughout the paper, \emph{twirled process matrix} means globally invariant twirled process matrix. This choice of representative is made because these are the process matrices that agents in a twirled world can reconstruct from local and non-local measurements. 
	
	As in the standard process-matrix framework, shared ancillas are allowed in $(\mathrm{O}1)$, while $(\mathrm{O}2)$ is imposed on complete local instruments.
	
	\begin{definition}[Twirled process matrix representatives]
		\label{def:twirled-W}
		Let the parties be $X^{(1)},\dots,X^{(n)}$. A $(G,U)$-twirled process matrix is a globally invariant representative $W\in L_{\mathrm{inv}}^{X^{(1)}_1X^{(1)}_2\cdots X^{(n)}_1X^{(n)}_2}$ such that: (i) $W$ satisfies $(\mathrm{O}1)$ for all local $(G,U)$-covariant instruments, allowing $(G,U)$-invariant shared ancillas; (ii) $W$ satisfies $(\mathrm{O}2)$ for all complete local $(G,U)$-covariant instruments.
	\end{definition}
	
	The characterisation below concerns finitely many finite-dimensional parties and these globally invariant representatives. Since local $(G,U)$-covariant CPTP maps form, in general, a proper subset of all local CPTP maps, the normalization constraints are weaker than in ordinary QT, so the corresponding set of representatives can be strictly larger.
	
	\begin{restatable}[Twirled process matrices]{theorem}{twprmat}
		\label{thm:tw-process-matrices}
		Let $(G,U)$ be a finite twirled world with parties $X^{(1)},\dots,X^{(n)}$. A matrix $W\in L^{X^{(1)}_1X^{(1)}_2\cdots X^{(n)}_1X^{(n)}_2}$ is an $n$-party twirled process matrix iff
		\begin{align}
			W&\in L_{\mathrm{inv}}^{X^{(1)}_1X^{(1)}_2\cdots X^{(n)}_1X^{(n)}_2}, \label{eq:tw-inv}\\
			W&\ge 0, \label{eq:tw-pos}\\
			\Tr\!\bigl[W(M_1\otimes\cdots\otimes M_n)\bigr]&=1 \label{eq:tw-norm}
		\end{align}
		for all Choi matrices $M_k$ of $(G,U)$-covariant CPTP maps.
	\end{restatable}
	
	Proofs are given in Appendix~\ref{app:characterisation-twirled-W}. Appendix~\ref{app:examples} presents the explicit process-matrix forms for one and two two-dimensional parties.
	
	\subsection{Swirled process matrices}
	\label{sec:swirled-pm}
	
	The same construction can be adapted to swirled quantum worlds. In the time-reversal-swirled case, namely real quantum theory (RQT), local instruments, ancillas, and operations are restricted to those with real density operators and real Choi matrices in the chosen basis~\cite{ying2025a}. A swirled process is therefore defined only up to operational indistinguishability with respect to the allowed RQT experiments as in the case of twirled process matrices.
	
	We fix as representative the one reconstructable within RQT, namely the real-symmetric one. The imaginary part and the real antisymmetric part are invisible to RQT experiments and are therefore removed by convention.
	
\begin{definition}[Swirled (RQT) process matrix representatives]
    \label{def:swirled-W}
    Let the parties be $X^{(1)},\dots,X^{(n)}$. A swirled (RQT) process matrix is a real-symmetric representative $W\in L_{\mathbb R}^{X^{(1)}_1X^{(1)}_2\cdots X^{(n)}_1X^{(n)}_2}$ such that: (i) $W$ satisfies $(\mathrm{O}1)$ for all local RQT instruments, allowing shared real ancillas; (ii) $W$ satisfies $(\mathrm{O}2)$ for all complete local RQT instruments.
\end{definition}

The same argument gives a full finite-dimensional characterisation for any number of parties. Here, as in Definition~\ref{def:swirled-W}, shared ancillas mean arbitrary shared real multipartite states; no decomposition into independent sources is imposed.

\begin{restatable}[Multipartite RQT process matrices]{theorem}{rqtprmat}

	\label{thm:rqt-process-matrices}
	Let the parties be $X^{(1)},\dots,X^{(n)}$. A matrix $W\in L^{X^{(1)}_1X^{(1)}_2\cdots X^{(n)}_1X^{(n)}_2}$ is an $n$-party RQT process matrix iff
	\begin{align}
		W &\in L_{\bR}^{X^{(1)}_1X^{(1)}_2\cdots X^{(n)}_1X^{(n)}_2}, \label{eq:rqt-real}\\
		W&\ge 0, \label{eq:rqt-pos}\\
		\Tr\!\bigl[W(M_1\otimes\cdots\otimes M_n)\bigr]&=1 \label{eq:rqt-norm}
	\end{align}
	for all CPTP Choi matrices $M_i\in L_{\bR}^{X^{(i)}_1X^{(i)}_2}$.
	Equivalently, condition~\eqref{eq:rqt-real} says that $W=\overline{W}$.
\end{restatable}

	Proofs are given in Appendix~\ref{app:characterisation-RQT-W}. Appendix~\ref{app:examples} presents the explicit process-matrix forms for one and two two-dimensional parties.

	\section{Characterisation of twirled process matrices}
	\label{app:characterisation-twirled-W}

	To handle condition (O1), it is convenient to package ancilla-assisted local tests into a single operator on the process-matrix space.
	
	Consider $n$ parties. For each party $X^{(k)}$, let $R_k$ be an ancilla carrying a unitary action $g\mapsto U_g^{R_k}$ of $G$, and let $\rho_{R_1\cdots R_n}$ be a shared $(G,U)$-invariant ancilla state. For each $k$, let
	\begin{align*}
		N_k\in L(\calH_{X^{(k)}_1}\otimes \calH_{R_k}\otimes \calH_{X^{(k)}_2})
	\end{align*}
	be the Choi operator of a local trace-nonincreasing $(G,U)$-covariant CP map
	\begin{align*}
		\mathcal N_k:L(\calH_{X^{(k)}_1}\otimes \calH_{R_k}) \to L(\calH_{X^{(k)}_2}).
	\end{align*}
	The corresponding effective test operator on the process-matrix legs is
	\begin{align*}
		T_{\mathrm{eff}}:=\Tr_{R_1\cdots R_n}\!\left[\bigl(I\otimes \rho_{R_1\cdots R_n}\bigr)\bigl(N_1\otimes\cdots\otimes N_n\bigr)\right].
	\end{align*}
	With this notation, condition (O1) says simply that
	\begin{align*}
		\TR{W T_{\eff}}\geq 0
	\end{align*}
	for every such test operator $T_{\eff}$.
	
	\subsection{Characterisation of test operators}
	
	\begin{lemma}
		\label{lem:twirled-effective-tester}
		Every test operator $T_{\eff}$ is positive,
		\begin{align*}
			T_{\eff} \geq 0,
		\end{align*}
		and globally invariant on the process-matrix legs. That is, for all $g\in G$,
		\begin{align*}
			\left(V_g^{(1)}\otimes\cdots\otimes V_g^{(n)}\right)	T_{\eff} \left(V_g^{(1)}\otimes\cdots\otimes V_g^{(n)}\right)^\dagger = T_{\mathrm{eff}},
		\end{align*}
		where
		\begin{align*}
			 V^{(k)}_g := U^{X^{(k)}_1}_g\otimes \overline{U}^{X^{(k)}_2}_g
		\end{align*}
		is the induced action on the Choi space of party $k$.
	\end{lemma}
	
	\begin{proof}
		Write
		\begin{align*}
			T_{\mathrm{eff}}=\Tr_R\!\left[\bigl(I_S\otimes \rho_{R_1\cdots R_n}\bigr)\bigl(N_1\otimes\cdots\otimes N_n\bigr)\right].
		\end{align*}
		
		We first prove positivity. Write
		\begin{align*}
			N:=N_1\otimes\cdots\otimes N_n,\qquad
			\rho:=\rho_{R_1\cdots R_n}.
		\end{align*}
		By cyclicity of the partial trace over the ancillary systems,
		\begin{align*}
			T_{\eff}
			&=
			\Tr_R\!\left[(I_S\otimes \rho)N\right]\\
			&=
			\Tr_R\!\left[(I_S\otimes \rho^{1/2})N(I_S\otimes \rho^{1/2})\right].
		\end{align*}
		The operator inside the partial trace is positive semidefinite, because $N\ge 0$. Since the partial trace is positive,
		\begin{align*}
			T_{\eff}\ge 0.
		\end{align*}
		
		We now prove invariance. Let
		\begin{align*}
			V_g:=V_g^{(1)}\otimes\cdots\otimes V_g^{(n)},
			\qquad
			W_g:=U_g^{R_1}\otimes\cdots\otimes U_g^{R_n}.
		\end{align*}
		After reordering tensor factors, covariance of the local Choi operators gives
		\begin{align*}
			(V_g\otimes W_g)\,(N_1\otimes\cdots\otimes N_n)\,(V_g\otimes W_g)^\dagger
			=
			N_1\otimes\cdots\otimes N_n,
		\end{align*}
		and invariance of the ancilla state gives
		\begin{align*}
			W_g\,\rho_{R_1\cdots R_n}\,W_g^\dagger=\rho_{R_1\cdots R_n}.
		\end{align*}
		Therefore
		\begin{align*}
			&V_g T_{\eff} V_g^\dagger=\Tr_R\!\left[(V_g\!\otimes\!W_g)\bigl(I_S\!\otimes\!\rho_{R_1\cdots R_n}\bigr)\bigl(N_1\!\otimes\cdots\otimes\!N_n\bigr)(V_g\!\otimes\!W_g)^\dagger\right] =
			\Tr_R\!\left[\bigl(I_S\otimes \rho_{R_1\cdots R_n}\bigr)\bigl(N_1\otimes\cdots\otimes N_n\bigr)\right] =
			T_{\eff}.
		\end{align*}
		So $T_{\mathrm{eff}}$ is globally invariant.
	\end{proof}
	
	For the next lemma, write $L_{\mathrm{inv}}(\calH)$ for the invariant subspace of $L(\calH)$, and $\mathrm{PSD}(L_{\mathrm{inv}}(\calH))$ for the set of positive semidefinite operators that lie in this subspace.
	
	\begin{lemma}
		\label{lem:aT}
		For every
		\begin{align*}
			T\in K:=\mathrm{PSD}(L^{X^{(1)}_1 X^{(1)}_2\cdots X^{(n)}_1 X^{(n)}_2}_{\mathrm{inv}}),
		\end{align*}
		there exists a scalar $a_T>0$ such that $a_T T$ is an effective twirled test operator generated by a $(G,U)$-invariant shared ancilla and local $(G,U)$-covariant trace-nonincreasing CP maps.
	\end{lemma}
	
	\begin{proof}
		Let $T\in K$. The case $T=0$ is trivial, so assume $T\neq 0$.
		
		We first realize $T$, up to an overall positive factor, as an effective test operator without worrying yet about covariance.
		For each party $k$, write
		\begin{align*}
			S_k:=X^{(k)}_1X^{(k)}_2,
		\end{align*}
		and choose an ancilla system $R'_k$ such that $R'_k\simeq S_k$. Define
		\begin{align*}
			\calH_S:=\bigotimes_{k=1}^n \calH_{X^{(k)}_1}\otimes \calH_{X^{(k)}_2},\qquad
			\calH_{SR}:=\calH_S\otimes\bigotimes_{k=1}^n \calH_{R'_k}.
		\end{align*}
		
		Consider the unnormalised maximally entangled vectors
		\begin{align*}
			\ket{\Omega}:=\bigotimes_{k=1}^{n}|\Omega_k\rangle
			:=
			\bigotimes_{k=1}^{n}\sum_{i,\alpha}|i\rangle_{X^{(k)}_1}\otimes|\alpha\rangle_{X^{(k)}_2}\otimes|i,\alpha\rangle_{R'_k},
		\end{align*}
		and define the ancilla state
		\begin{align*}
			\rho_{R'}:=\frac{T^T}{\TR{T}}.
		\end{align*}
		This is a valid state because $T\ge 0$ implies $T^T\ge 0$.
		
		For Hermitian operators $A,B$ on
\begin{align*}
	\bigotimes_{k=1}^n \calH_{X^{(k)}_1}\otimes \calH_{X^{(k)}_2},
\end{align*}
one has
\begin{align*}
	\TR{(A\otimes B)\ket{\Omega}\bra{\Omega}}=\TR{AB^T}.
\end{align*}
Hence, for every Hermitian $A$,
\begin{align*}
	\TR{(A\otimes\rho_{R'})\ket{\Omega}\bra{\Omega}}
	=
	\TR{A\rho_{R'}^T}
	=
	\frac{\TR{AT}}{\TR{T}}.
\end{align*}
Equivalently,
\begin{align*}
	\TR{AT}
	&=
	\TR{T}\,
	\TR{(A\otimes\rho_{R'})\ket{\Omega}\bra{\Omega}}\\
	&=
	\TR{T}\,
	\TR{(A\otimes\rho_{R'})
	(\ket{\Omega_1}\bra{\Omega_1}\otimes \dots \otimes \ket{\Omega_n}\bra{\Omega_n})}\\
	&=
	\TR{T}\left(\prod_{k=1}^{n} c_k\right)
	\TR{(A\otimes\rho_{R'})
	\left(\frac{\ket{\Omega_1}\bra{\Omega_1}}{c_1}\otimes \dots \otimes
	\frac{\ket{\Omega_n}\bra{\Omega_n}}{c_n}\right)},
\end{align*}
with $c_k \ge d_{X^{(k)}_1}$.
Since
\begin{align*}
	N'_k:=\frac{\ket{\Omega_k}\bra{\Omega_k}}{c_k}
\end{align*}
is the Choi operator of a trace-nonincreasing map, defining
\begin{align*}
	a_T:=\frac{1}{\TR{T}\prod_{k=1}^{n}c_k},
\end{align*}
shows that $a_TT$ has the form of an effective test operator.
		
		We now modify the construction so that it is covariant. For each party $k$, let $G_k$ be a classical register with orthonormal basis $\{|g\rangle\}_{g\in G}$, on which $G$ acts by
		\begin{align*}
			L_h|g\rangle = |hg\rangle.
		\end{align*}
		Set
		\begin{align*}
			R_k:=G_k\otimes R'_k,
		\end{align*}
		The register $R'_k$ carries the trivial representation; the nontrivial action on $R_k$ is only the left-regular action on $G_k$.
		Define
		\begin{align*}
			V_g^{(k)}:=U_g^{X^{(k)}_1}\otimes \overline{U}_g^{X^{(k)}_2}.
		\end{align*}
		
		Let
		\begin{align*}
			\sigma_{R_1\cdots R_n}:=\omega_{G_1\cdots G_n}\otimes \rho_{R'_1\cdots R'_n},
		\end{align*}
		where
		\begin{align*}
			\omega_{G_1\cdots G_n}:=\frac{1}{|G|}\sum_{g\in G}|g\rangle\langle g|_{G_1}\otimes\cdots\otimes |g\rangle\langle g|_{G_n}.
		\end{align*}
		This state is $(G,U)$-invariant.
		
		Define the local Choi operators
		\begin{align*}
			\widehat N_k
			:=
			\sum_{g\in G}|g\rangle\langle g|_{G_k}\otimes \bigl(V_g^{(k)}\otimes I_{R'_k}\bigr)\,N'_k\,\bigl(V_g^{(k)}\otimes I_{R'_k}\bigr)^\dagger.
		\end{align*}
		Each $\widehat N_k$ is positive semidefinite. It is also trace-nonincreasing, since
		\begin{align*}
			&\Tr_{X^{(k)}_2}\widehat N_k= \sum_{g\in G} |g\rangle\langle g|_{G_k}\otimes \bigl(U_g^{X^{(k)}_1}\otimes I_{R'_k}\bigr)\bigl(\Tr_{X^{(k)}_2}N'_k\bigr)\bigl(U_g^{X^{(k)}_1}\otimes I_{R'_k}\bigr)^\dagger \le
			I_{G_k}\otimes I_{X^{(k)}_1R'_k}.
		\end{align*}
		So $\widehat N_k$ is the Choi operator of a valid local trace-nonincreasing CP map
		\begin{align*}
			\widehat{\mathcal N}_k:L(\calH_{X^{(k)}_1}\otimes \calH_{G_k}\otimes \calH_{R'_k}) \to L(\calH_{X^{(k)}_2}),
		\end{align*}
		and this map is $(G,U)$-covariant.
		
		Let $\widehat T$ be the effective operator produced by the invariant ancilla $\sigma_{R_1\cdots R_n}$ and the local covariant maps $\widehat N_1,\dots,\widehat N_n$. Writing
		\begin{align*}
			V_g:=V_g^{(1)}\otimes\cdots\otimes V_g^{(n)},
		\end{align*}
		we find, for every Hermitian $A$ on $\calH_S$,
		\begin{align*}
			\Tr(A\widehat T)
			&=
			\Tr\!\left[(A\otimes \sigma_{R_1\cdots R_n})(\widehat N_1\otimes\cdots\otimes \widehat N_n)\right] =\frac{1}{|G|}\sum_{g\in G}\Tr\!\left[\bigl(V_g^\dagger A V_g \otimes \rho_{R'_1\cdots R'_n}\bigr)(N'_1\otimes\cdots\otimes N'_n)\right] =\\
			&=\frac{1}{|G|}\sum_{g\in G} a_T\,\Tr\!\left(V_g^\dagger A V_g\,T\right)=
			a_T\,\Tr\!\left[A\,\frac{1}{|G|}\sum_{g\in G} V_g T V_g^\dagger\right].
		\end{align*}
		Since $T$ is invariant,
		\begin{align*}
			V_g T V_g^\dagger = T \qquad \forall g\in G,
		\end{align*}
		so
		\begin{align*}
			\Tr(A\widehat T)=a_T\,\Tr(AT)\qquad \forall A.
		\end{align*}
		Hence
		\begin{align*}
			\widehat T=a_TT.
		\end{align*}
		So $a_TT$ is indeed an effective twirled test operator.
	\end{proof}
	
	Before proving the characterisation theorem, we record a simple positivity lemma.
	
	\begin{lemma}[Positivity from invariant positive test operators]
		\label{lem:inv-positivity}
		Let $ W\in L_{\mathrm{inv}}(\calH)$ be Hermitian. If
		\begin{align*}
			\Tr(WT)\ge 0 \qquad \forall\,T\in L_{\mathrm{inv}}(\calH)\ \text{with}\ T\ge 0,
		\end{align*}
		then
		\begin{align*}
			W\ge 0.
		\end{align*}
	\end{lemma}
	
	\begin{proof}
		Assume for contradiction that $W\not\ge 0$. Since $W$ is Hermitian, there exist a unit vector $|\psi\rangle\in \calH$ and an eigenvalue $\lambda<0$ such that
		\begin{align*}
			W|\psi\rangle=\lambda|\psi\rangle.
		\end{align*}
		Because $W$ is invariant, we have
		\begin{align*}
			V_gWV_g^\dagger=W \qquad \forall g\in G,
		\end{align*}
		or equivalently
		\begin{align*}
			V_gW=WV_g\qquad \forall g\in G.
		\end{align*}
		Therefore, for every $g\in G$,
		\begin{align*}
			W(V_g|\psi\rangle)=V_gW|\psi\rangle=\lambda\,V_g|\psi\rangle.
		\end{align*}
		
		Now define
		\begin{align*}
			T:=\frac{1}{|G|}\sum_{g\in G}	V_g|\psi\rangle\langle\psi|V_g^\dagger.
		\end{align*}
		Each term is positive semidefinite, so $T\ge 0$. Also, $T$ is invariant:
		\begin{align*}
			V_hTV_h^\dagger
			&=
			\frac{1}{|G|}\sum_{g\in G} V_hV_g|\psi\rangle\langle\psi|V_g^\dagger V_h^\dagger=
			\frac{1}{|G|}\sum_{g\in G} V_{hg}|\psi\rangle\langle\psi|V_{hg}^\dagger
			=
			T.
		\end{align*}
		So $T\in L_{\mathrm{inv}}(\calH)$ and $T\ge 0$.
		
		By hypothesis,
		\begin{align*}
			\Tr(WT)\ge 0.
		\end{align*}
		But
		\begin{align*}
			\Tr(WT)
			&=
			\frac{1}{|G|}\sum_{g\in G}\Tr\!\left(WV_g|\psi\rangle\langle\psi|V_g^\dagger\right) =
			\frac{1}{|G|}\sum_{g\in G}\langle\psi|V_g^\dagger W V_g|\psi\rangle =
			\frac{1}{|G|}\sum_{g\in G}\langle\psi|W|\psi\rangle =
			\lambda<0,
		\end{align*}
		a contradiction. Hence $W\ge 0$.
	\end{proof}
	
	\subsection{Characterisation of process matrices for twirled worlds}
	
    \twprmat*
	\begin{proof}
		\noindent $(\Rightarrow)$ Assume that $W$ is an $n$-party twirled process matrix. By Definition~\ref{def:twirled-W}, $W$ is invariant, which gives Eq.~\eqref{eq:tw-inv}. Applying (O2) to complete local $(G,U)$-covariant instruments gives Eq.~\eqref{eq:tw-norm}.
		
		It remains to prove Eq.~\eqref{eq:tw-pos}, namely $W\ge 0$. We first show that $W$ is Hermitian.
		
		Let $T$ be any effective twirled test operator. Since $T\ge 0$, it is Hermitian, and since $\Tr(WT)$ is a probability, it is real. Therefore
		\begin{align*}
			\Tr\!\bigl[(W-W^\dagger)T\bigr]
			=
			\Tr(WT)-\overline{\Tr(WT)}
			=
			0
		\end{align*}
		for every effective twirled test operator $T$.
		
Now let $A\in L^{X^{(1)}_1 X^{(1)}_2\cdots X^{(n)}_1 X^{(n)}_2}_{\mathrm{inv}}$ be any Hermitian operator, and write $A=A_+-A_-$ with
		\begin{align*}
			A_\pm\in K:=\mathrm{PSD}(L^{X^{(1)}_1 X^{(1)}_2\cdots X^{(n)}_1 X^{(n)}_2}_{\mathrm{inv}}).
		\end{align*}
		By Lemma~\ref{lem:aT}, there exist positive scalars $a_{A_+},a_{A_-}>0$ such that $a_{A_+}A_+$ and $a_{A_-}A_-$ are effective twirled test operators. Hence
		\begin{align*}
			\Tr\!\bigl[(W-W^\dagger)A_+\bigr]=0,
			\qquad
			\Tr\!\bigl[(W-W^\dagger)A_-\bigr]=0,
		\end{align*}
		and therefore
		\begin{align*}
			\Tr\!\bigl[(W-W^\dagger)A\bigr]=0
		\end{align*}
		for every invariant Hermitian $A$.
		
		Since $W$ is invariant, so is
		\begin{align*}
			H:=\frac{W-W^\dagger}{2i},
		\end{align*}
		and $H$ is Hermitian. Taking $A=H$ gives
		\begin{align*}
			0=\Tr\!\bigl[(W-W^\dagger)H\bigr]=2i\,\Tr(H^2).
		\end{align*}
		Hence $\Tr(H^2)=0$, so $H=0$. Therefore $W=W^\dagger$.
		
		Now let $T\in K$ be arbitrary. By Lemma~\ref{lem:aT}, there exists $a_T>0$ such that $a_TT$ is an effective twirled test operator. Applying (O1),
		\begin{align*}
			0\le \Tr\!\bigl(W(a_TT)\bigr)=a_T\,\Tr(WT).
		\end{align*}
		Since $a_T>0$, this implies
		\begin{align*}
			\Tr(WT)\ge 0 \qquad \forall\, T\in K.
		\end{align*}
		Because $W$ is Hermitian, invariant, and nonnegative on all invariant positive operators, Lemma~\ref{lem:inv-positivity} applies and yields
		\begin{align*}
			W\ge 0.
		\end{align*}
		This proves Eq.~\eqref{eq:tw-pos}.
		
		\medskip
		\noindent $(\Leftarrow)$ Assume now that
		\begin{align*}
			W&\in L_{\mathrm{inv}}^{X^{(1)}_1X^{(1)}_2\cdots X^{(n)}_1X^{(n)}_2}, \\
			W&\ge 0,\\
			\TR{W(M_1\otimes\cdots\otimes M_n)}&=1
		\end{align*}
		for all Choi matrices $M_k$ of $(G,U)$-covariant channels.
		
		Let $T$ be any effective twirled test operator. By Lemma~\ref{lem:twirled-effective-tester}, $T$ is positive semidefinite and invariant, hence $T\in K$. Since $W\ge 0$ and $T\ge 0$,
		\begin{align*}
			\Tr(WT)=\Tr\!\bigl(W^{1/2}TW^{1/2}\bigr)\ge 0.
		\end{align*}
		So condition (O1) holds.
		
		Condition (O2) is exactly the assumed normalization condition, since complete local twirled instruments are built from local $(G,U)$-covariant CPTP maps. Hence (O2) also holds.
		
		Therefore $W$ is an $n$-party twirled process matrix.
	\end{proof}

	\section{Characterisation of process matrices for RQT}
	\label{app:characterisation-RQT-W}
	
\begin{lemma}[Effective RQT test operators are real and positive]
	\label{lem:rqt-effective-testers-real-pos}
	Any effective RQT test operator generated by a shared real positive semidefinite ancilla state together with real positive semidefinite local Choi operators is itself real and positive semidefinite:
	\begin{align*}
		T=\overline{T},\qquad T\ge 0.
	\end{align*}
	In particular, this holds for arbitrary multipartite ancilla-assisted RQT testers in the process framework.
\end{lemma}

\begin{proof}
	Write
	\begin{align*}
		N:=N_1\otimes\cdots\otimes N_n,\qquad
		\rho:=\rho_{R_1\cdots R_n},
	\end{align*}
	where the $N_i$ are the real positive semidefinite local Choi operators and $\rho$ is the shared real positive semidefinite ancilla state. The effective test operator is
	\begin{align*}
		T=\Tr_R\!\left[(I\otimes \rho)N\right].
	\end{align*}
	By cyclicity of the partial trace over the ancillary systems,
	\begin{align*}
		T
		=
		\Tr_R\!\left[(I\otimes \rho^{1/2})N(I\otimes \rho^{1/2})\right].
	\end{align*}
	Since $N\ge 0$, the operator inside the partial trace is positive semidefinite, and therefore $T\ge 0$. Moreover, $\rho^{1/2}$ is real because $\rho$ is real positive semidefinite, and all $N_i$ are real. Hence the operator inside the partial trace is real, and so is its partial trace. Thus
	\begin{align*}
		T=\overline{T},\qquad T\ge 0.
	\end{align*}
\end{proof}
	
	\begin{lemma}[Every real positive operator is an effective multipartite RQT tester up to scale]
	\label{lem:aT-rqt}
Let the parties be $X^{(1)},\dots,X^{(n)}$. For every
\begin{align*}
    T\in K_{\mathbb R}:=\mathrm{PSD}\!\left(
    L_{\mathbb R}^{X^{(1)}_1X^{(1)}_2\cdots X^{(n)}_1X^{(n)}_2}
    \right),
\end{align*}
	there exists a scalar $a_T>0$ such that $a_TT$ is an effective $n$-party RQT test operator generated by a shared real ancilla state and local real trace-nonincreasing CP maps.
\end{lemma}

\begin{proof}
	Let $T\in K_{\mathbb R}$. If $T=0$, the claim is trivial, so assume $T\neq 0$.
	
	For each party $X^{(i)}$, write
	\begin{align*}
		S_i:=X^{(i)}_1X^{(i)}_2.
	\end{align*}
	Choose ancillas
	\begin{align*}
		R'_i\simeq S_i,
	\end{align*}
	with fixed bases identified with the product bases of $S_i$.
	
	For each $i$, define the unnormalised maximally correlated vector
	\begin{align*}
		|\Omega_i\rangle
		:=
		\sum_{p_i,q_i}
		|p_i\rangle_{X^{(i)}_1}\otimes |p_i,q_i\rangle_{R'_i}\otimes |q_i\rangle_{X^{(i)}_2}.
	\end{align*}
	Choose
	\begin{align*}
		0<c_i\le \frac{1}{d_{X^{(i)}_1}},
	\end{align*}
	and set
	\begin{align*}
		N'_i:=c_i\,|\Omega_i\rangle\langle\Omega_i|.
	\end{align*}
	Since $|\Omega_i\rangle$ has only real coefficients in the chosen basis, $N'_i$ is real positive semidefinite. Moreover,
	\begin{align*}
		\Tr_{X^{(i)}_2}N'_i
		=
		c_i\sum_{q_i}
		\left(\sum_{p_i}|p_i\rangle_{X^{(i)}_1}|p_i,q_i\rangle_{R'_i}\right)
		\left(\sum_{p'_i}\langle p'_i|_{X^{(i)}_1}\langle p'_i,q_i|_{R'_i}\right).
	\end{align*}
	The vectors in parentheses are mutually orthogonal for different $q_i$, and each has squared norm $d_{X^{(i)}_1}$. Hence the largest eigenvalue of $\Tr_{X^{(i)}_2}N'_i$ is $c_i d_{X^{(i)}_1}\le 1$, so
	\begin{align*}
		\Tr_{X^{(i)}_2}N'_i\le I_{X^{(i)}_1R'_i}.
	\end{align*}
	Thus $N'_i$ is a valid local trace-nonincreasing RQT Choi operator.
	
	Now define the shared ancilla state
	\begin{align*}
		\rho_{R'_1\cdots R'_n}:=\frac{T^T}{\Tr(T)},
	\end{align*}
	where we identify $R'_1\cdots R'_n$ with $S_1\cdots S_n$ in the fixed product basis. Since $T\ge 0$, also $T^T\ge 0$, and since $T$ is real, $\rho_{R'_1\cdots R'_n}$ is a real density operator.
	
	For every Hermitian operator $Y\in L^{X^{(1)}_1X^{(1)}_2\cdots X^{(n)}_1X^{(n)}_2}$, the maximally correlated contraction gives
	\begin{align*}
\Tr\!\left[ (Y\otimes \rho_{R'_1\cdots R'_n})(N'_1\otimes\cdots\otimes N'_n)\right]=\frac{\prod_{i=1}^n c_i}{\Tr(T)}\,\Tr(YT).
	\end{align*}
	Therefore, defining
	\begin{align*}
		a_T:=\frac{\prod_{i=1}^n c_i}{\Tr(T)},
	\end{align*}
	the effective $n$-party RQT test operator generated by $\rho_{R'_1\cdots R'_n}$ and the local Choi operators $N'_1,\dots,N'_n$ is $a_TT$.
	
	Hence $a_TT$ is an effective $n$-party RQT test operator.
\end{proof}
	
\rqtprmat*
\begin{proof}
$(\Rightarrow)$ Assume that $W$ is an $n$-party RQT process matrix in the sense of Definition~\ref{def:swirled-W}. By definition, $W$ is the real-symmetric representative, hence
\begin{align*}
    W\in L_{\bR}^{X^{(1)}_1X^{(1)}_2\cdots X^{(n)}_1X^{(n)}_2}.
\end{align*}
We first prove positivity. Let
\begin{align*}
    T\in K_{\mathbb R}:=\mathrm{PSD}\!\left(
    L_{\bR}^{X^{(1)}_1X^{(1)}_2\cdots X^{(n)}_1X^{(n)}_2}
    \right).
\end{align*}
By Lemma~\ref{lem:aT-rqt}, there exists $a_T>0$ such that $a_TT$ is an effective RQT test operator generated by a shared real ancilla state and local real trace-nonincreasing CP maps. Applying condition $(\mathrm O1)$ to this test operator gives $0\le \Tr[W\,a_TT]=a_T\Tr[WT]$, and therefore $\Tr[WT]\ge 0$ for all $T\in K_{\mathbb R}$.

Since $W$ is real symmetric, it is enough to test positivity against real positive semidefinite operators. Indeed, for every vector $|\psi\rangle$, let $T_\psi:=|\psi\rangle\!\langle\psi|+\overline{|\psi\rangle\!\langle\psi|}$. Then $T_\psi\in K_{\mathbb R}$, and
\begin{align*}
    \Tr[WT_\psi]
    &=
    \langle\psi|W|\psi\rangle
    +
    \Tr\!\left[W\,\overline{|\psi\rangle\!\langle\psi|}\right]  \\
    &=
    \langle\psi|W|\psi\rangle
    +
    \overline{\langle\psi|W|\psi\rangle}
    =
    2\langle\psi|W|\psi\rangle,
\end{align*}
where we used that $W$ is real symmetric, hence Hermitian. Thus $\langle\psi|W|\psi\rangle\ge 0$ for all $|\psi\rangle$, and therefore $W\ge 0$.

It remains to prove normalization. Applying condition $(\mathrm O2)$ to one-outcome complete local RQT instruments, with CPTP Choi matrices $M_i\in L_{\bR}^{X^{(i)}_1X^{(i)}_2}$, gives
\begin{align*}
    \Tr\!\bigl[W(M_1\otimes\cdots\otimes M_n)\bigr]=1
\end{align*}
for all CPTP Choi matrices $M_i\in L_{\bR}^{X^{(i)}_1X^{(i)}_2}$. This proves the \enquote{only if} direction.

$(\Leftarrow)$ Conversely, assume that
\begin{align*}
    W\in L_{\bR}^{X^{(1)}_1X^{(1)}_2\cdots X^{(n)}_1X^{(n)}_2},
    \qquad
    W\ge 0,
\end{align*}
and that $\Tr\!\bigl[W(M_1\otimes\cdots\otimes M_n)\bigr]=1$ for all CPTP Choi matrices $M_i\in L_{\bR}^{X^{(i)}_1X^{(i)}_2}$. Thus $W$ is a real-symmetric representative, as required in Definition~\ref{def:swirled-W}.

We verify $(\mathrm O1)$. Consider an arbitrary effective RQT test operator generated by a shared real ancilla state and local real trace-nonincreasing CP maps. Such a test operator is positive semidefinite and belongs to $L_{\bR}^{X^{(1)}_1X^{(1)}_2\cdots X^{(n)}_1X^{(n)}_2}$. Since $W\ge 0$, its pairing with this positive test operator is nonnegative. Hence $(\mathrm O1)$ holds.

We verify $(\mathrm O2)$. Let $\{\mathcal M^{(i)}_{a_i}\}_{a_i}$ be complete local RQT instruments for each party $X^{(i)}$, with Choi matrices $M^{(i)}_{a_i}\in L_{\bR}^{X^{(i)}_1X^{(i)}_2}$. For each $i$, the sum $M_i:=\sum_{a_i}M^{(i)}_{a_i}$ is a CPTP Choi matrix in $L_{\bR}^{X^{(i)}_1X^{(i)}_2}$. Therefore, by the assumed normalization condition,
\begin{align*}
    \sum_{a_1,\dots,a_n} \Tr\!\left[ W\left(M^{(1)}_{a_1}\otimes\cdots\otimes M^{(n)}_{a_n}\right) \right]
    =\Tr\!\left[W\left(M_1\otimes\cdots\otimes M_n\right)\right] =1.
\end{align*}
Thus $(\mathrm O2)$ holds. Hence $W$ is an $n$-party RQT process matrix.
\end{proof}

	\section{Examples}
	\label{app:examples}
	
	\subsection{Single 2-dimensional party}	
	
	\paragraph{Ordinary quantum theory}
	\label{sec:qt-example}
	
	From~\cite{oreshkov2012a} the single-party process matrix in ordinary quantum theory takes the form
	\begin{align}\label{eq:single-QT}
		W_{\mathrm{QT}}&=\tfrac{1}{2}\Big(\id+v_x \sigma_x^{A_1}+v_y \sigma_y^{A_1}+v_z \sigma_z^{A_1}\Big),
	\end{align}
	with $v_x,v_y,v_z\in\mathbb R$, subject to $W_{\mathrm{QT}}\ge 0$.
	
	Thus, in the single-party case, ordinary QT process matrices are density operators on the input space $A_1$, tensored with the maximally mixed operator on the output space $A_2$. 
	
	\paragraph{Parity-twirled fermionic world}
	\label{sec:parity-example}
	
	The parity-twirled fermionic world~\cite{centeno2025} has $G=\mathbb Z_2$, with $U_\pi=\sigma_z$ on a single fermionic mode. The Pauli basis transforms as $\sigma_z\sigma_\mu\sigma_z=p_\mu\sigma_\mu$, with $p_0=p_z=+1$ and $p_x=p_y=-1$. Thus, on a single mode, the invariant operators are spanned by $\{\id,\sigma_z\}$. On the input-output pair $A_1A_2$, the globally invariant operator space is spanned by $\id$, $\sigma_z^{A_1}$, $\sigma_z^{A_2}$, $\sigma_z^{A_1}\sigma_z^{A_2}$, $\sigma_x^{A_1}\sigma_x^{A_2}$, $\sigma_x^{A_1}\sigma_y^{A_2}$, $\sigma_y^{A_1}\sigma_x^{A_2}$, and $\sigma_y^{A_1}\sigma_y^{A_2}$.
	
	A parity-covariant CPTP Choi matrix is therefore a linear combination of $\id$, $\sigma_z^{A_2}$, $\sigma_z^{A_1}\sigma_z^{A_2}$, $\sigma_x^{A_1}\sigma_x^{A_2}$, $\sigma_x^{A_1}\sigma_y^{A_2}$, $\sigma_y^{A_1}\sigma_x^{A_2}$, and $\sigma_y^{A_1}\sigma_y^{A_2}$, subject to positivity. Since we now require the process matrix itself to be globally invariant, $W_{\mathrm{par}}$ (the subscript \enquote{par} indicates that it is a process matrix for parity-twirled QT) must lie in the invariant subspace above. Imposing the normalization condition $\Tr[W_{\mathrm{par}}M_{\mathrm{par}}]=1$ for all parity-covariant CPTP maps $M_{\mathrm{par}}$ removes every invariant component except the identity and the input term $\sigma_z^{A_1}$. Therefore the single-party parity-twirled process matrix has the form
	\begin{equation}\label{eq:single-par}
		W_{\mathrm{par}}=\tfrac{1}{2}\bigl(\id+v_z\sigma_z^{A_1}\bigr),
	\end{equation}
	with $v_z\in\mathbb R$, subject to $W_{\mathrm{par}}\ge 0$.
	
	Thus the single-party parity-twirled process is a parity-twirled fermionic qubit state on the input system, mirroring the ordinary QT case where a single-party process matrix is just a qubit state.
	
	\paragraph{Real quantum theory}
	\label{sec:rqt-example}
	
	For qubits, complex conjugation acts as $\overline{\sigma_\mu}=s_\mu\sigma_\mu$, with $s_0=s_x=s_z=+1$ and $s_y=-1$. Hence an operator is real iff its expansion contains no term with an odd number of $\sigma_y$ factors. On a single input-output pair $A_1A_2$, the real operator space is spanned by $\id$, $\sigma_x^{A_1}$, $\sigma_z^{A_1}$, $\sigma_x^{A_2}$, $\sigma_z^{A_2}$, $\sigma_x^{A_1}\sigma_x^{A_2}$, $\sigma_x^{A_1}\sigma_z^{A_2}$, $\sigma_z^{A_1}\sigma_x^{A_2}$, $\sigma_z^{A_1}\sigma_z^{A_2}$, and $\sigma_y^{A_1}\sigma_y^{A_2}$.
	
	A real CPTP Choi matrix is therefore a linear combination of $\id$, $\sigma_x^{A_2}$, $\sigma_z^{A_2}$, $\sigma_x^{A_1}\sigma_x^{A_2}$, $\sigma_x^{A_1}\sigma_z^{A_2}$, $\sigma_z^{A_1}\sigma_x^{A_2}$, $\sigma_z^{A_1}\sigma_z^{A_2}$, and $\sigma_y^{A_1}\sigma_y^{A_2}$, subject to positivity. Since we now require the process matrix itself to be real, $W_{\mathrm{RQT}}$ must lie in the real subspace above. Imposing the normalization condition $\Tr[W_{\mathrm{RQT}}M_{\mathrm{RQT}}]=1$ for all real CPTP Choi matrices $M_{\mathrm{RQT}}$ removes every real component except the identity and the input terms $\sigma_x^{A_1}$ and $\sigma_z^{A_1}$. Therefore the single-party RQT process matrix has the form
	\begin{equation}\label{eq:single-rqt}
		W_{\mathrm{RQT}}=\tfrac{1}{2}\bigl(\id+v_x\sigma_x^{A_1}+v_z\sigma_z^{A_1}\bigr),
	\end{equation}
	with $v_x,v_z\in\mathbb R$, subject to $W_{\mathrm{RQT}}\ge 0$.
	
	Thus the single-party RQT process is a real qubit state on the input system, again mirroring the ordinary QT case and showing that the construction is consistent at the single-party level.
	
	\subsection{Two 2-dimensional parties}
	
	\paragraph{Ordinary quantum theory: two parties}
	\label{sec:qt-example-bipartite}
	
	For two two-dimensional parties, the most general bipartite qubit process matrix is a Hermitian operator on $A_1A_2B_1B_2$ satisfying the OCB positivity and normalization conditions. In the Pauli basis, the allowed support types are $1$, $A_1$, $B_1$, $A_1B_1$, $A_2B_1$, $A_1B_2$, $A_1A_2B_1$, and $A_1B_1B_2$. Thus one can write
	\begin{equation}
		\begin{aligned}
			W_{\mathrm{QT}}=\frac{1}{4}\Big(&\id+\sum_{i\in\{x,y,z\}} a_i\,\sigma_i^{A_1}+\sum_{j\in\{x,y,z\}} b_j\,\sigma_j^{B_1}+\sum_{i,j\in\{x,y,z\}} c_{ij}\,\sigma_i^{A_1}\sigma_j^{B_1}+\sum_{i,j\in\{x,y,z\}} d_{ij}\,\sigma_i^{A_2}\sigma_j^{B_1}\\
			&+\sum_{i,j\in\{x,y,z\}} e_{ij}\,\sigma_i^{A_1}\sigma_j^{B_2}+\sum_{i,j,k\in\{x,y,z\}} f_{ijk}\,\sigma_i^{A_1}\sigma_j^{A_2}\sigma_k^{B_1}+\sum_{i,j,k\in\{x,y,z\}} g_{ijk}\,\sigma_i^{A_1}\sigma_j^{B_1}\sigma_k^{B_2}\Big),
		\end{aligned}
		\label{eq:bipartite-qt}
	\end{equation}
	with all coefficients real, subject to $W_{\mathrm{QT}}\ge 0$.
	
	Thus the bipartite QT process matrix is the ordinary two-qubit OCB process matrix. OCB-forbidden terms as postselection terms, local-loop terms, channel-with-local-loop terms, and global-loop terms are absent.
	
	\paragraph{Parity-twirled fermionic world: two parties}
	\label{sec:parity-example-bipartite}
	
	For each local input-output pair $X_1X_2$, with $X=A,B$, the parity-invariant Pauli operators are $\id\id$, $\sigma_z^{X_1}$, $\sigma_z^{X_2}$, $\sigma_z^{X_1}\sigma_z^{X_2}$, $\sigma_x^{X_1}\sigma_x^{X_2}$, $\sigma_x^{X_1}\sigma_y^{X_2}$, $\sigma_y^{X_1}\sigma_x^{X_2}$, and $\sigma_y^{X_1}\sigma_y^{X_2}$. Since the process matrix itself must be globally invariant, the bipartite process matrix is a linear combination of globally invariant Pauli strings. Imposing the normalization conditions against all local parity-covariant CPTP maps leaves the usual OCB-visible sector, together with an additional globally invariant but locally inaccessible sector. Thus
	\begin{equation}
		\begin{aligned}
			W_{\mathrm{par}}=\frac{1}{4}\Big(&\id+a\,\sigma_z^{A_1}+b\,\sigma_z^{B_1}+c\,\sigma_z^{A_1}\sigma_z^{B_1}+d\,\sigma_z^{A_2}\sigma_z^{B_1}+e\,\sigma_z^{A_1}\sigma_z^{B_2}+\sum_{P\in\mathcal I_A^{(2)}} \alpha_P\,P\,\sigma_z^{B_1}+\sum_{Q\in\mathcal I_B^{(2)}} \beta_Q\,\sigma_z^{A_1}\,Q+\Delta_{\mathrm{par}}\Big),
		\end{aligned}
		\label{eq:bipartite-par}
	\end{equation}
	where
	\begin{align*}
		\mathcal I_A^{(2)}&=\{\sigma_z^{A_1}\sigma_z^{A_2}, \sigma_x^{A_1}\sigma_x^{A_2}, \sigma_x^{A_1}\sigma_y^{A_2}, \sigma_y^{A_1}\sigma_x^{A_2}, \sigma_y^{A_1}\sigma_y^{A_2}\},\\
		\mathcal I_B^{(2)}&=\{\sigma_z^{B_1}\sigma_z^{B_2}, \sigma_x^{B_1}\sigma_x^{B_2}, \sigma_x^{B_1}\sigma_y^{B_2}, \sigma_y^{B_1}\sigma_x^{B_2}, \sigma_y^{B_1}\sigma_y^{B_2}\},
	\end{align*}
	and
	\begin{equation}
		\Delta_{\mathrm{par}}=\sum_{P\in\mathcal O_A}\sum_{Q\in\mathcal O_B} h_{PQ}\,P\,Q ,
	\end{equation}
	with
	\begin{align*}
		\mathcal O_A&=\{\sigma_x^{A_1},\sigma_y^{A_1}, \sigma_x^{A_2},\sigma_y^{A_2}, \sigma_x^{A_1}\sigma_z^{A_2}, \sigma_y^{A_1}\sigma_z^{A_2}, \sigma_z^{A_1}\sigma_x^{A_2}, \sigma_z^{A_1}\sigma_y^{A_2}\},\\
		\mathcal O_B&=\{\sigma_x^{B_1},\sigma_y^{B_1}, \sigma_x^{B_2},\sigma_y^{B_2}, \sigma_x^{B_1}\sigma_z^{B_2}, \sigma_y^{B_1}\sigma_z^{B_2},\sigma_z^{B_1}\sigma_x^{B_2}, \sigma_z^{B_1}\sigma_y^{B_2}\},
	\end{align*}
	and all coefficients are real, subject to $W_{\mathrm{par}}\ge 0$.
	
	The terms different from $\Delta_{\mathrm{par}}$ in Eq.~\eqref{eq:bipartite-par} are the parity-restricted analogue of the ordinary OCB sector and are all OCB-allowed: they realize the support types $1$, $A_1$, $B_1$, $A_1B_1$, $A_2B_1$, $A_1B_2$, $A_1A_2B_1$, and $A_1B_1B_2$, but only with parity-invariant local factors. The term $\Delta_{\mathrm{par}}$ contains the additional globally invariant directions that are not reconstructable from local covariant data. This sector is locally inaccessible; it may contain OCB-forbidden support types. Thus the visible part of $W_{\mathrm{par}}$ still mirrors the ordinary bipartite qubit process matrix, while the extra parity-twirled structure is local tomographically hidden.
	
	\paragraph{Real quantum theory: two parties}
	\label{sec:rqt-example-bipartite}
	
	For each local input-output pair $X_1X_2$, with $X=A,B$, the real Pauli operators are
	\begin{align*}
		\id\id,\quad \sigma_x^{X_1},\quad \sigma_z^{X_1},\quad \sigma_x^{X_2},\quad \sigma_z^{X_2},
	\end{align*}
	\begin{align*}
		\sigma_x^{X_1}\sigma_x^{X_2},\quad \sigma_x^{X_1}\sigma_z^{X_2},\quad \sigma_z^{X_1}\sigma_x^{X_2},\quad \sigma_z^{X_1}\sigma_z^{X_2},\quad \sigma_y^{X_1}\sigma_y^{X_2}.
	\end{align*}
	Thus the real self-adjoint local space is
    \begin{align*}
    \Herm_{\mathbb R}^{X_1X_2}= \mathrm{span}_{\mathbb R}\!\big\{ & \id\id, \sigma_x^{X_1}, \sigma_z^{X_1}, \sigma_x^{X_2}, \sigma_z^{X_2}, \sigma_x^{X_1}\sigma_x^{X_2},\sigma_x^{X_1}\sigma_z^{X_2}, \sigma_z^{X_1}\sigma_x^{X_2}, \sigma_z^{X_1}\sigma_z^{X_2},\sigma_y^{X_1}\sigma_y^{X_2} \big\}.
    \end{align*}
	
	Define
	\begin{align*}
		K_X:=\mathrm{span}_{\mathbb R}\!\left\{ I,\ \sigma_x^{X_1},\ \sigma_z^{X_1} \right\}, \  K_X^0:=\mathrm{span}_{\mathbb R}\!\left\{ \sigma_x^{X_1},\ \sigma_z^{X_1} \right\},
	\end{align*}
	and
	\begin{align*}
		T_X:=\mathrm{span}_{\mathbb R}\!\big\{& \sigma_x^{X_2},\sigma_z^{X_2}, \sigma_x^{X_1}\sigma_x^{X_2}, \sigma_x^{X_1}\sigma_z^{X_2}, \sigma_z^{X_1}\sigma_x^{X_2},\sigma_z^{X_1}\sigma_z^{X_2}, \sigma_y^{X_1}\sigma_y^{X_2} \big\}.
	\end{align*}
	Then
	\begin{align*}
		\Herm_{\mathbb R}^{X_1X_2}=K_X\oplus T_X,
	\end{align*}
	and, by Lemma~\ref{lem:rqt-visible-sector-qubits}, the locally accessible part of any bipartite two-qubit RQT process matrix belongs to
	\begin{align*}
		(K_A\otimes K_B)\oplus(T_A\otimes K_B^0)\oplus(K_A^0\otimes T_B).
	\end{align*}
	Accordingly, every bipartite two-qubit RQT process matrix can be decomposed as
	\begin{align*}
		W_{\mathrm{RQT}}=W_{\mathrm{vis}}+\Delta_{\mathrm{RQT}},
	\end{align*}
	where $W_{\mathrm{vis}}$ is the locally accessible component
	\begin{equation}
		\begin{aligned}
			W_{\mathrm{vis}}=\frac{1}{4}\Big(&\id+\sum_{i\in\{x,z\}} a_i\,\sigma_i^{A_1}+\sum_{j\in\{x,z\}} b_j\,\sigma_j^{B_1}+\sum_{i,j\in\{x,z\}} c_{ij}\,\sigma_i^{A_1}\sigma_j^{B_1}+\sum_{t\in\mathcal T_A}\sum_{j\in\{x,z\}} d_{tj}\,t\,\sigma_j^{B_1}+\sum_{i\in\{x,z\}}\sum_{t\in\mathcal T_B} e_{it}\,\sigma_i^{A_1}\,t\Big),
		\end{aligned}
		\label{eq:bipartite-rqt}
	\end{equation}
	with
	\begin{align*}
		\mathcal T_X= \big\{& \sigma_x^{X_2}, \sigma_z^{X_2}, \sigma_x^{X_1}\sigma_x^{X_2}, \sigma_x^{X_1}\sigma_z^{X_2}, \sigma_z^{X_1}\sigma_x^{X_2},\sigma_z^{X_1}\sigma_z^{X_2}, \sigma_y^{X_1}\sigma_y^{X_2} \big\},
	\end{align*}
    all coefficients real, and $\Delta_{\mathrm{RQT}}$ orthogonal to
    \begin{align*}
        \Herm_{\mathbb R}^{A_1A_2}\otimes \Herm_{\mathbb R}^{B_1B_2}.
    \end{align*}
	

	\section{Causal loops and tomographic locality}
	\label{app:causal-loops}
	
	Before proving the main theorem of the section we need to introduce some notation and a short lemma that will result useful in the case of RQT.
	
For each qubit input-output pair $X_1X_2$, with $X=A,B$, the real self-adjoint local space introduced above is
\begin{align*}
    \Herm_{\mathbb{R}}^{X_1X_2}
    =\mathrm{span}_{\mathbb{R}}\big\{&
    \id,\sigma_x^{X_1},\sigma_z^{X_1},
    \sigma_x^{X_2},\sigma_z^{X_2},
    \sigma_x^{X_1}\sigma_x^{X_2},
    \sigma_x^{X_1}\sigma_z^{X_2},
    \sigma_z^{X_1}\sigma_x^{X_2},
    \sigma_z^{X_1}\sigma_z^{X_2},
    \sigma_y^{X_1}\sigma_y^{X_2}
    \big\}.
\end{align*}
	Define
	\begin{align*}
		K_X&:=\mathrm{span}_{\mathbb{R}}\{\id,\sigma_x^{X_1},\sigma_z^{X_1}\},\\
		K_X^0&:=\mathrm{span}_{\mathbb{R}}\{\sigma_x^{X_1},\sigma_z^{X_1}\},\\
		T_X&:=\{N\in \Herm_{\mathbb{R}}^{X_1X_2}:\Tr_{X_2}N=0\}.
	\end{align*}
	
	\begin{lemma}[Visible sector of bipartite qubit RQT]
		\label{lem:rqt-visible-sector-qubits}
		Let $W_{\mathrm{la}}\in \Herm_{\mathbb{R}}^{A_1A_2}\otimes \Herm_{\mathbb{R}}^{B_1B_2}$ be a real self-adjoint operator such that
		\begin{align*}
			\Tr\bigl[W_{\mathrm{la}}(M^{A_1A_2}\otimes M^{B_1B_2})\bigr]=1
		\end{align*}
		for all real CPTP Choi matrices $M^{A_1A_2}\in \Herm_{\mathbb R}^{A_1A_2}$ and $M^{B_1B_2}\in \Herm_{\mathbb R}^{B_1B_2}$. Then
		\begin{align*}
			W_{\mathrm{la}}\in (K_A\otimes K_B)\oplus(T_A\otimes K_B^0)\oplus(K_A^0\otimes T_B).
		\end{align*}
	\end{lemma}
	
	\begin{proof}
		For each $X=A,B$, the displayed basis of $\Herm_{\mathbb{R}}^{X_1X_2}$ is a real Pauli basis. Since $\Tr(\sigma_\mu\sigma_\nu)=2\delta_{\mu\nu}$, distinct Pauli strings are orthogonal when using the Hilbert--Schmidt inner product. Hence $K_X$ and $T_X$ are orthogonal, and $\Herm_{\mathbb{R}}^{X_1X_2}=K_X\oplus T_X$.
		
		Now let $M^{X_1X_2}$ be a real CPTP Choi matrix. Expanding it in the basis of $\Herm_{\mathbb{R}}^{X_1X_2}$ and imposing $\Tr_{X_2}M^{X_1X_2}=I_{X_1}$ shows that its $K_X$-component is $\tfrac12\id$, with no $\sigma_x^{X_1}$ or $\sigma_z^{X_1}$ term. Therefore every real CPTP Choi matrix has a unique decomposition
		\begin{align*}
			M^{X_1X_2}=\tfrac12\id+N_X,\qquad N_X\in T_X.
		\end{align*}
		Conversely, for every $N_X\in T_X$, the operator $\tfrac12\id+\varepsilon N_X$ is a real CPTP Choi matrix for all sufficiently small real $\varepsilon$, since $\tfrac12\id$ is strictly positive and the trace-preserving condition is automatic for $N_X\in T_X$.
		
		Take arbitrary $N_A\in T_A$ and $N_B\in T_B$. For sufficiently small $\varepsilon,\delta\in\mathbb{R}$, the operators
		\begin{align*}
			M_A(\varepsilon)&:=\tfrac12\id+\varepsilon N_A,\\
			M_B(\delta)&:=\tfrac12\id+\delta N_B
		\end{align*}
		are real CPTP Choi matrices, so by hypothesis
		\begin{align*}
			\Tr\bigl[W_{\mathrm{la}}(M_A(\varepsilon)\otimes M_B(\delta))\bigr]=1
		\end{align*}
		for all sufficiently small $\varepsilon,\delta$. Expanding, we get
		\begin{align*}
			1&=\Tr\bigl[W_{\mathrm{la}}((\tfrac12\id+\varepsilon N_A)\otimes(\tfrac12\id+\delta N_B))\bigr]\\ &= \Tr\bigl[W_{\mathrm{la}}(\tfrac12\id\otimes\tfrac12\id)\bigr]+\frac{\varepsilon}{2}\Tr\bigl[W_{\mathrm{la}}(N_A\otimes\id)\bigr]+\frac{\delta}{2}\Tr\bigl[W_{\mathrm{la}}(\id\otimes N_B)\bigr]+\varepsilon\delta\,\Tr\bigl[W_{\mathrm{la}}(N_A\otimes N_B)\bigr].
		\end{align*}
		Since this identity holds for all sufficiently small $\varepsilon,\delta$, it follows that
		\begin{align*}
			\Tr\bigl[W_{\mathrm{la}}(N_A\otimes\id)\bigr]&=0\qquad\forall\,N_A\in T_A,\\
			\Tr\bigl[W_{\mathrm{la}}(\id\otimes N_B)\bigr]&=0\qquad\forall\,N_B\in T_B,\\
			\Tr\bigl[W_{\mathrm{la}}(N_A\otimes N_B)\bigr]&=0\qquad\forall\,N_A\in T_A,\ \forall\,N_B\in T_B.
		\end{align*}
		
	Using $\Herm_{\mathbb{R}}^{A_1A_2}=K_A\oplus T_A$ and $\Herm_{\mathbb{R}}^{B_1B_2}=K_B\oplus T_B$, we obtain the orthogonal decomposition
    \begin{align*}
            \Herm_{\mathbb{R}}^{A_1A_2}\otimes \Herm_{\mathbb{R}}^{B_1B_2}&=(K_A\otimes K_B)\oplus(K_A\otimes T_B)\oplus(T_A\otimes K_B)\oplus(T_A\otimes T_B)\\
			&=(K_A\otimes K_B)\oplus(\id\otimes T_B)\oplus(K_A^0\otimes T_B)\oplus(T_A\otimes\id)\oplus(T_A\otimes K_B^0)\oplus(T_A\otimes T_B).
		\end{align*}
		The three orthogonality relations above imply that the components of $W_{\mathrm{la}}$ along $T_A\otimes\id$, $\id\otimes T_B$, and $T_A\otimes T_B$ vanish. Therefore
		\begin{align*}
			W_{\mathrm{la}}\in (K_A\otimes K_B)\oplus(T_A\otimes K_B^0)\oplus(K_A^0\otimes T_B).
		\end{align*}
	\end{proof}
	
	Thus the normalised locally accessible sector contains only the support types
	\begin{align*}
		&1,\ A_1,\ B_1,\ A_1B_1,\\
		&A_2B_1,\ A_1B_2,\ A_1A_2B_1,\ A_1B_1B_2,
	\end{align*}
	that is, only OCB-allowed terms.
	
	We introduce some notation before the next theorem. Let $\Herm_{\mathrm{sym}}^{X_1X_2}$ denote the locally accessible self-adjoint symmetric operator space, namely $\Herm_{\mathrm{inv}}^{X_1X_2}$ in a twirled world and $\Herm_{\mathbb{R}}^{X_1X_2}$ in RQT. Furthermore we denote by $W_{\mathrm{la}}^{\mathrm{sym}}$ the projection of the process matrix onto the locally accessible self-adjoint subspace $\bigotimes_i \Herm_{\mathrm{sym}}^{X^{(i)}_1X^{(i)}_2}$.
	
	\begin{theorem}[Local symmetry hides OCB-forbidden terms in the bipartite qubit case]
		\label{thm:sym-ocb-forbidden-invisible}
		Consider a bipartite process matrix with two-dimensional local input and output systems, in either a finite bipartite twirled world or bipartite real quantum theory. Then the locally reconstructable part of the process matrix contains only OCB-allowed support types. Equivalently, every OCB-forbidden term is invisible to local symmetric tomography.
	\end{theorem}
	\begin{proof}
		We work throughout in the bipartite qubit setting. For each party $X=A,B$, let $S_X:=\Herm_{\mathrm{sym}}^{X_1X_2}$, where $\Herm_{\mathrm{sym}}^{X_1X_2}=\Herm_{\mathrm{inv}}^{X_1X_2}$ in a twirled world and $\Herm_{\mathrm{sym}}^{X_1X_2}=\Herm_{\mathbb{R}}^{X_1X_2}$ in RQT. Define $S:=S_A\otimes S_B$, let $P_S$ be the orthogonal projector onto $S$, and write
		\begin{align*}
			W_{\mathrm{la}}^{\mathrm{sym}}:=P_S(W)
		\end{align*}
		for the locally accessible part of the process matrix.
		
		For each party $X=A,B$, the allowed local symmetric Choi operators span the full real Hermitian space $S_X$: indeed, every Hermitian $H\in S_X$ can be written as a difference of two positive operators in $S_X$, and after a sufficiently small rescaling these become allowed trace-nonincreasing local outcomes. Hence
		\begin{align*}
			\mathrm{span}_{\mathbb{R}}\{M_A\otimes M_B\}=S_A\otimes S_B=S,
		\end{align*}
		where $M_A$ and $M_B$ range over the allowed local symmetric Choi operators.
		
		We now show that $W_{\mathrm{la}}^{\mathrm{sym}}=P_S(W)$ is the part of $W$ reconstructed by local symmetric tomography. Since every allowed local symmetric Choi operator lies in the corresponding space $S_X$, every allowed product tester $M_A\otimes M_B$ lies in $S$. Hence
		\begin{align*}
			\Tr\bigl[W(M_A\otimes M_B)\bigr]=\Tr\bigl[P_S(W)(M_A\otimes M_B)\bigr].
		\end{align*}
		Because the allowed product testers span $S$, all locally accessible probabilities depend only on $P_S(W)$, so local symmetric tomography reconstructs $W_{\mathrm{la}}^{\mathrm{sym}}$.
		
		We now treat separately the twirled and RQT cases.
		
		\medskip
		\noindent\textbf{Twirled world.}  We use here the stronger operator-level construction proved in Section~\ref{app:correlation-sets}, in the proof of Theorem~\ref{thm:twirl-subset}.
		
		By the operator-level construction used in the proof of Theorem~\ref{thm:twirl-subset} in Section~\ref{app:correlation-sets}, there exists an ordinary QT process matrix $\widetilde W$ on the same local spaces such that
		\begin{align*}
			\Tr\bigl[W(M_A\otimes M_B)\bigr]=\Tr\bigl[\widetilde W(M_A\otimes M_B)\bigr]
		\end{align*}
		for all allowed local symmetric Choi operators $M_A$ and $M_B$. Since these testers span $S$, it follows that
		\begin{align*}
			P_S(W)=P_S(\widetilde W).
		\end{align*}
		
		Now
		\begin{align*}
			S=\Herm_{\mathrm{inv}}^{A_1A_2}\otimes \Herm_{\mathrm{inv}}^{B_1B_2},
		\end{align*}
		so, on Hermitian operators,
    \begin{align*}
        P_S=\mathcal T_A\otimes \mathcal T_B,
    \end{align*}
		where
    \begin{align*}
    \mathcal T_A(X)&:=\frac{1}{|G|}\sum_{g\in G}V_g^A X(V_g^A)^\dagger,
    \qquad
    V_g^A:= U_g^{A_1}\otimes \overline{U}_g^{A_2},\\
    \mathcal T_B(Y)&:=\frac{1}{|G|}\sum_{g\in G}V_g^B Y(V_g^B)^\dagger,
    \qquad
    V_g^B:=U_g^{B_1}\otimes \overline{U}_g^{B_2}.
\end{align*}
		
		For each party $X=A,B$, we decompose the local operator space into the four Pauli-support sectors $1$, $X_1$, $X_2$, and $X_1X_2$. Conjugation by $V_g^X=U_g^{X_1}\otimes \overline{U}_g^{X_2}$ preserves each of these sectors: $I$ is fixed, and conjugation on each tensor factor preserves tracelessness. Therefore the local twirling maps $\mathcal T_A$ and $\mathcal T_B$ are block-diagonal with respect to this decomposition, and so is $P_S=\mathcal T_A\otimes \mathcal T_B$. In particular, $P_S$ cannot move weight from an OCB-allowed support sector to an OCB-forbidden one.
		
		Since $\widetilde W$ is an ordinary bipartite qubit QT process matrix, its Hilbert--Schmidt expansion has support only on the OCB-allowed sectors. Hence the same is true for $P_S(\widetilde W)$. Using
		\begin{align*}
			W_{\mathrm{la}}^{\mathrm{sym}}=P_S(W)=P_S(\widetilde W),
		\end{align*}
		we conclude that the locally reconstructable part of a bipartite qubit twirled-world process matrix contains only OCB-allowed support types. Equivalently, every OCB-forbidden term is invisible to local symmetric tomography in the bipartite qubit twirled case.
		
		\medskip
		\noindent\textbf{RQT.} Here
		\begin{align*}
			S=\Herm_{\mathbb{R}}^{A_1A_2}\otimes \Herm_{\mathbb{R}}^{B_1B_2},
		\end{align*}
		and by the argument above the locally accessible part is again
		\begin{align*}
			W_{\mathrm{la}}^{\mathrm{sym}}=P_S(W).
		\end{align*}
		Moreover, since all local real Choi operators lie in $S$, we have
		\begin{align*}
			\Tr\bigl[W(M^{A_1A_2}\otimes M^{B_1B_2})\bigr]=\Tr\bigl[P_S(W)(M^{A_1A_2}\otimes M^{B_1B_2})\bigr]
		\end{align*}
		for all local real Choi operators. In particular,
		\begin{align*}
			\Tr\bigl[W_{\mathrm{la}}^{\mathrm{sym}}(M^{A_1A_2}\otimes M^{B_1B_2})\bigr]=1
		\end{align*}
		for all local real CPTP Choi matrices.
		
		Lemma~\ref{lem:rqt-visible-sector-qubits} therefore applies to $W_{\mathrm{la}}^{\mathrm{sym}}$ and gives
		\begin{align*}
			W_{\mathrm{la}}^{\mathrm{sym}}\in (K_A\otimes K_B)\oplus(T_A\otimes K_B^0)\oplus(K_A^0\otimes T_B).
		\end{align*}
		Hence its support contains only the OCB-allowed bipartite qubit support types
		\begin{align*}
			&1,\quad A_1,\quad B_1,\quad A_1B_1,\quad A_2B_1,\quad A_1B_2,\quad A_1A_2B_1,\quad A_1B_1B_2.
		\end{align*}
		Thus every OCB-forbidden term is invisible to local real tomography in the bipartite qubit RQT case.
	\end{proof}

	\section{Operational consequences and correlation sets}
	\label{app:correlation-sets}
	
	In this Section we prove the correlation results stated in the main text. The proof of Theorem~\ref{thm:simulation} adapts the incoherent simulation of Ref.~\cite{ying2025a} to the process-matrix setting. The key extra step is to show that, starting from a QT process matrix $W$, one can construct a twirled process matrix $\widehat W$ in the finite-dimensional $G$-twirled theory that reproduces the same correlation.
	
		\begin{theorem}[Fixed-representation inclusion]
		\label{thm:twirl-subset}
		For every finite group $G$, every finite local dimension tuple $d$ and every $n$-party twirled world $(G,U)$,
		\begin{align*}
			\mathcal C_{(G,U)}^{\proc}(d)
			\subseteq
			\mathcal C_{\mathrm{QT}}^{\proc}(d).
		\end{align*}
	\end{theorem}
	\begin{proof}
		For each party $X^{(k)}$, let
		\begin{align*}
            \mathfrak U_g^{(k)}(M):=(U_g^{X^{(k)}_1}\otimes \overline{U}_g^{X^{(k)}_2}) M (U_g^{X^{(k)}_1}\otimes \overline{U}_g^{X^{(k)}_2})^\dagger
		\end{align*}
		be the induced action on local Choi operators. Given an $n$-party twirled process matrix $W$, define
		\begin{align*}
			\widetilde W
			:=
			\frac{1}{|G|^n}\sum_{g_1,\dots,g_n\in G}
			(\mathfrak U_{g_1}^{(1)}\otimes\cdots\otimes\mathfrak U_{g_n}^{(n)})(W).
		\end{align*}
		By Theorem~\ref{thm:tw-process-matrices}, $W\ge0$. Since each $\mathfrak U_g^{(k)}$ acts by unitary conjugation, $\widetilde W\ge 0$.
		
		Let $T_k$ be arbitrary CPTP Choi operators on the local spaces of party $X^{(k)}$, and define
		\begin{align*}
			N_k:= \frac{1}{|G|}\sum_{g\in G}\mathfrak U_g^{(k)}(T_k).
		\end{align*}
		Then each $N_k$ is a $(G,U)$-covariant CPTP Choi operator. Indeed, positivity is immediate, while
		\begin{align*}
	\Tr_{X^{(k)}_2}N_k&=\frac{1}{|G|}\sum_{g\in G}U_g^{X^{(k)}_1}\bigl(\Tr_{X^{(k)}_2}T_k\bigr)(U_g^{X^{(k)}_1})^\dagger =\frac{1}{|G|}\sum_{g\in G}U_g^{X^{(k)}_1}\,\id_{X^{(k)}_1}\,(U_g^{X^{(k)}_1})^\dagger=\id_{X^{(k)}_1}.
		\end{align*}
		Using the change of variables \(g\mapsto g^{-1}\) in the normalized finite-group sums,
		\begin{align*}
			\Tr\!\bigl[
			\widetilde W(T_1\otimes\cdots\otimes T_n)
			\bigr]
			&=
			\Tr\!\bigl[
			W(N_1\otimes\cdots\otimes N_n)
			\bigr] \\
			&=1.
		\end{align*}
		Hence $\widetilde W$ is an ordinary QT process matrix on the same local spaces.
		
		Finally, let $M_k$ be local $(G,U)$-covariant Choi operators. Since $\mathfrak U_g^{(k)}(M_k)=M_k$ for all $g\in G$,
		\begin{align*}
			\Tr\!\bigl[
			\widetilde W(M_1\otimes\cdots\otimes M_n)
			\bigr]
			=
			\Tr\!\bigl[
			W(M_1\otimes\cdots\otimes M_n)
			\bigr].
		\end{align*}
		Thus $\widetilde W$ reproduces all symmetric product probabilities of $W$. Replacing $W$ by $\widetilde W$ and keeping the same local instruments therefore gives the same correlation in QT. Hence
		\begin{align*}
			\mathcal C_{(G,U)}^{\proc}(d)
			\subseteq
			\mathcal C_{\mathrm{QT}}^{\proc}(d).
		\end{align*}
	\end{proof}
	
	We now prove the converse inclusion. Let $G$ be finite, and let $\mathsf{Tw}_G$ be the finite-dimensional $G$-twirled theory introduced in the main text. By definition, every finite-dimensional Hilbert space carrying a unitary representation of $G$ is a system of $\mathsf{Tw}_G$.
	
	Fix a local input-output pair $X=(X_1,X_2)$. Let
	\begin{align*}
		R_{X_1}\simeq R_{X_2}\simeq \mathbb C^{|G|}
	\end{align*}
	with orthonormal basis $\{\ket g\}_{g\in G}$, and let $G$ act on this basis by
	\begin{align*}
		L_h\ket g=\ket{hg}.
	\end{align*}
	Equip $X_1$ and $X_2$ with the trivial representation and define
	\begin{align*}
		\widehat X_1:=X_1\otimes R_{X_1},
		\qquad
		\widehat X_2:=X_2\otimes R_{X_2},
	\end{align*}
	with local action
	\begin{align*}
		\widehat U_h^{X_1}:=\id_{X_1}\otimes L_h,
		\qquad
		\widehat U_h^{X_2}:=\id_{X_2}\otimes L_h.
	\end{align*}
	These enlarged systems again belong to $\mathsf{Tw}_G$.
	
	Let $\widehat{\mathfrak U}_h^X$ denote the induced action on $L^{\widehat X_1\widehat X_2}$. After the canonical reordering of tensor factors, define
	\begin{align*}
		\Gamma_X(M)
		:=
		\sum_{g\in G}
		\ketbra{g}{g}_{R_{X_1}}
		\otimes
		\ketbra{g}{g}_{R_{X_2}}
		\otimes
		M
	\end{align*}
	and
	\begin{align*}
		\mathcal D_X(\widehat M)
		:=
		\frac{1}{|G|}
		\sum_{g\in G}
		\Tr_{R_{X_2}}
		\Big[
		(\bra g_{R_{X_1}}\otimes \id)\,
		\widehat M\,
		(\ket g_{R_{X_1}}\otimes \id)
		\Big].
	\end{align*}
	These are the Choi-operator versions of the local simulation maps of Ref.~\cite{ying2025a} that we call encoding and decoding map respectively.
	
	\begin{lemma}
		\label{lem:gamma-delta-properties}
		For every local input-output pair $X=(X_1,X_2)$, the maps $\Gamma_X$ and $\mathcal D_X$ are completely positive. Moreover:
		\begin{enumerate}[label=(\roman*)]
			\item if $M$ is a CPTP Choi operator on $X_1X_2$, then $\Gamma_X(M)$ is a $(G,\widehat U)$-covariant CPTP Choi operator on $\widehat X_1\widehat X_2$;
			\item if $\widehat T$ is a CPTP Choi operator on $\widehat X_1\widehat X_2$, then $\mathcal D_X(\widehat T)$ is a CPTP Choi operator on $X_1X_2$;
			\item one has
			\begin{align*}
				\mathcal D_X\circ \Gamma_X=\id;
			\end{align*}
			\item the Hilbert--Schmidt adjoint $\mathcal D_X^\dagger$ is completely positive and has globally $(G,\widehat U)$-invariant image.
		\end{enumerate}
	\end{lemma}
	
	\begin{proof}
		Complete positivity of $\Gamma_X$ is immediate, since it is a sum of maps of the form
		\begin{align*}
			M\mapsto
			\ketbra{g}{g}_{R_{X_1}}
			\otimes
			\ketbra{g}{g}_{R_{X_2}}
			\otimes
			M.
		\end{align*}
		Complete positivity of $\mathcal D_X$ is also immediate, since it is a sum of compression maps followed by a partial trace.
		
		If $M$ is CPTP, then $\Gamma_X(M)\ge 0$ and
		\begin{align*}
			\Tr_{\widehat X_2}\Gamma_X(M)
			&=
			\sum_{g\in G}
			\ketbra{g}{g}_{R_{X_1}}
			\otimes
			\Tr_{X_2}M =
			\sum_{g\in G}
			\ketbra{g}{g}_{R_{X_1}}
			\otimes
			\id_{X_1}
			=
			\id_{\widehat X_1},
		\end{align*}
		so $\Gamma_X(M)$ is CPTP. Moreover,
		\begin{align*}
			\widehat{\mathfrak U}_h^X(\Gamma_X(M))
			&=
			\sum_{g\in G}
			\ketbra{hg}{hg}_{R_{X_1}}
			\otimes
			\ketbra{hg}{hg}_{R_{X_2}}
			\otimes
			M =
			\Gamma_X(M),
		\end{align*}
		so $\Gamma_X(M)$ is $(G,\widehat U)$-covariant. This proves (i).
		
		If $\widehat T$ is CPTP, then $\mathcal D_X(\widehat T)\ge 0$ and
		\begin{align*}
			\Tr_{X_2}\mathcal D_X(\widehat T)
			&=
			\frac{1}{|G|}
			\sum_{g\in G}
			\bra g_{R_{X_1}}
			\Tr_{\widehat X_2}\widehat T
			\ket g_{R_{X_1}} =
			\frac{1}{|G|}
			\sum_{g\in G}
			\bra g_{R_{X_1}}
			\id_{X_1R_{X_1}}
			\ket g_{R_{X_1}}
			=
			\id_{X_1},
		\end{align*}
		so $\mathcal D_X(\widehat T)$ is CPTP. This proves (ii).
		
		For every $M\in L^{X_1X_2}$,
		\begin{align*}
			(\bra g_{R_{X_1}}\otimes \id)\,
			\Gamma_X(M)\,
			(\ket g_{R_{X_1}}\otimes \id)
			=
			\ketbra{g}{g}_{R_{X_2}}\otimes M,
		\end{align*}
		and therefore
		\begin{align*}
			\mathcal D_X(\Gamma_X(M))
			&=
			\frac{1}{|G|}
			\sum_{g\in G}
			\Tr_{R_{X_2}}
			\bigl[
			\ketbra{g}{g}_{R_{X_2}}\otimes M
			\bigr]
			=
			M.
		\end{align*}
		This proves (iii).
		
		For every $h\in G$ and every $\widehat M\in L^{\widehat X_1\widehat X_2}$,
		\begin{align*}
			\mathcal D_X\!\bigl(\widehat{\mathfrak U}_h^X(\widehat M)\bigr)
			&=
			\frac{1}{|G|}
			\sum_{g\in G}
			\Tr_{R_{X_2}}
			\Big[
			(\bra g\otimes \id)\,
			\widehat{\mathfrak U}_h^X(\widehat M)\,
			(\ket g\otimes \id)
			\Big] =
			\frac{1}{|G|}
			\sum_{g\in G}
			\Tr_{R_{X_2}}
			\Big[
			(\bra{h^{-1}g}\otimes \id)\,
			\widehat M\,
			(\ket{h^{-1}g}\otimes \id)
			\Big] =
			\mathcal D_X(\widehat M),
		\end{align*}
		where we relabel $g\mapsto h^{-1}g$ in the last step. Hence
		\begin{align*}
			\mathcal D_X\circ \widehat{\mathfrak U}_h^X=\mathcal D_X
		\end{align*}
		for all $h\in G$. Passing to Hilbert--Schmidt adjoints gives
		\begin{align*}
			\widehat{\mathfrak U}_h^X\circ \mathcal D_X^\dagger=\mathcal D_X^\dagger,
		\end{align*}
		so the image of $\mathcal D_X^\dagger$ is globally invariant. Since $\mathcal D_X$ is completely positive on a finite-dimensional matrix algebra, so is $\mathcal D_X^\dagger$. This proves (iv).
	\end{proof}

		\begin{theorem}[Finite-dimensional simulation theorem]
		\label{thm:simulation}
		Let $G$ be a finite group. Then every finite-dimensional QT process-matrix correlation can be realized in the finite-dimensional $G$-twirled theory $\mathsf{Tw}_G$ after adjoining suitable finite-dimensional ancillas to the local systems. Equivalently,
		\begin{align*}
			\mathcal C_{\mathrm{QT}}^{\proc,\fin}
			\subseteq
			\mathcal C_{\mathsf{Tw}_G}^{\proc,\fin}.
		\end{align*}
	\end{theorem}
	\begin{proof}
		Fix a finite-dimensional QT process-matrix realization
		\begin{align*}
			p(a_1,\dots,a_n|x_1,\dots,x_n)
			=
			\Tr\!\Bigl[
			W\bigl(
			M^{(1)}_{a_1|x_1}\otimes\cdots\otimes M^{(n)}_{a_n|x_n}
			\bigr)
			\Bigr],
		\end{align*}
		where $W$ is an $n$-party QT process matrix and, for each party $k$ and setting $x_k$, the family $\{M^{(k)}_{a_k|x_k}\}_{a_k}$ is a local quantum instrument.
		
		For each party $k$, regard the original spaces as systems $X^{(k)}_1$ and $X^{(k)}_2$ of $\mathsf{Tw}_G$ carrying the trivial representation, and add ancillas
		\begin{align*}
			R^{(k)}_1\simeq R^{(k)}_2\simeq \mathbb C^{|G|}
		\end{align*}
		with basis $\{\ket g\}_{g\in G}$ and action
		\begin{align*}
			L_h\ket g=\ket{hg}.
		\end{align*}
		Define
		\begin{align*}
			\widehat X^{(k)}_1:=X^{(k)}_1\otimes R^{(k)}_1,
			\qquad
			\widehat X^{(k)}_2:=X^{(k)}_2\otimes R^{(k)}_2.
		\end{align*}
		These enlarged systems again belong to $\mathsf{Tw}_G$.
		
		Let $\Gamma_k$ and $\mathcal D_k$ be the corresponding encoding and decoding maps from Lemma~\ref{lem:gamma-delta-properties}, and define
		\begin{align*}
			\widehat W
			:=
			(\mathcal D_1^\dagger\otimes\cdots\otimes \mathcal D_n^\dagger)(W).
		\end{align*}
		Then $\widehat W\ge 0$, since each $\mathcal D_k^\dagger$ is completely positive, and $\widehat W$ is globally invariant by Lemma~\ref{lem:gamma-delta-properties}(iv).
		
		Let $\widehat T_1,\dots,\widehat T_n$ be arbitrary local $(G,\widehat U)$-covariant CPTP Choi operators. By Lemma~\ref{lem:gamma-delta-properties}(ii), each $\mathcal D_k(\widehat T_k)$ is an ordinary CPTP Choi operator, so
		\begin{align*}
			\Tr\!\bigl[
			\widehat W(\widehat T_1\otimes\cdots\otimes \widehat T_n)
			\bigr]
			&=
			\Tr\!\Bigl[
			W\bigl(
			\mathcal D_1(\widehat T_1)\otimes\cdots\otimes \mathcal D_n(\widehat T_n)
			\bigr)
			\Bigr] \\
			&=1.
		\end{align*}
		By Theorem~\ref{thm:tw-process-matrices}, $\widehat W$ is therefore a valid twirled process matrix in $\mathsf{Tw}_G$.
		
		For each party $k$ and setting $x_k$, Lemma~\ref{lem:gamma-delta-properties}(i) shows that the operators $\Gamma_k(M^{(k)}_{a_k|x_k})$ are completely positive and that their sum over $a_k$ is a $(G,\widehat U)$-covariant CPTP Choi operator. Thus the encoded families define valid local instruments in $\mathsf{Tw}_G$.
		
		Finally, using adjointness and Lemma~\ref{lem:gamma-delta-properties}(iii),
		\begin{align*}
			\Tr\!\Bigl[
			\widehat W\bigl(
			\Gamma_1(M^{(1)}_{a_1|x_1})\otimes\cdots\otimes \Gamma_n(M^{(n)}_{a_n|x_n})
			\bigr)
			\Bigr] &=
			\Tr\!\Bigl[
			W\bigl(
			\mathcal D_1(\Gamma_1(M^{(1)}_{a_1|x_1}))\otimes\cdots\otimes
			\mathcal D_n(\Gamma_n(M^{(n)}_{a_n|x_n}))
			\bigr)
			\Bigr]\\ & =
			\Tr\!\Bigl[
			W\bigl(
			M^{(1)}_{a_1|x_1}\otimes\cdots\otimes M^{(n)}_{a_n|x_n}
			\bigr)
			\Bigr] =
			p(a_1,\dots,a_n|x_1,\dots,x_n).
		\end{align*}
		Hence the original QT correlation is reproduced in the same finite-dimensional $G$-twirled theory, proving
		\begin{align*}
			\mathcal C_{\mathrm{QT}}^{\proc,\fin}
			\subseteq
			\mathcal C_{\mathsf{Tw}_G}^{\proc,\fin}.
		\end{align*}
	\end{proof}
	
		\begin{theorem}[Finite-dimensional operational equivalence]
		\label{thm:findimequi}
		Let $G$ be a finite group. Then
		\begin{align*}
			\mathcal C_{\mathsf{Tw}_G}^{\proc,\fin}
			=
			\mathcal C_{\mathrm{QT}}^{\proc,\fin}.
		\end{align*}
	\end{theorem}
	\begin{proof}
		Every finite realization in $\mathsf{Tw}_G$ uses some finite-dimensional local representation tuple $(G,U)$ on some finite local dimension tuple $d$. By Theorem~\ref{thm:twirl-subset},
		\begin{align*}
			\mathcal C_{(G,U)}^{\proc}(d)
			\subseteq
			\mathcal C_{\mathrm{QT}}^{\proc}(d).
		\end{align*}
		Taking the union over all finite choices of local systems in $\mathsf{Tw}_G$ gives
		\begin{align*}
			\mathcal C_{\mathsf{Tw}_G}^{\proc,\fin}
			\subseteq
			\mathcal C_{\mathrm{QT}}^{\proc,\fin}.
		\end{align*}
		The reverse inclusion is exactly Theorem~\ref{thm:simulation}. Hence
		\begin{align*}
			\mathcal C_{\mathsf{Tw}_G}^{\proc,\fin}
			=
			\mathcal C_{\mathrm{QT}}^{\proc,\fin}. 
		\end{align*}
	\end{proof}

	\begin{theorem}[Strict inclusion of QT in RQT process correlations]
	\label{thm:QT-strict-RQT-supp}
	For finite-dimensional bipartite process-matrix correlations,
	\begin{align*}
		\mathcal C_{\mathrm{QT}}^{\proc,\fin}
		\subsetneq
		\mathcal C_{\mathrm{RQT}}^{\proc,\fin}.
	\end{align*}
\end{theorem}
	\begin{proof}
		Part of this proof is inspired by Refs.~\cite{stueckelberg1960,mckague2009,ying2025a}. We first prove the inclusion. In order to do so we are going to need two main ingredients.
		
		The first ingredient we need is the encoding map $\enc_X$. To define it, we first introduce the realification map $\rea$. To every complex Hilbert space $\cH_X\simeq\bC^d$, with the basis in which complex conjugation is fixed, we associate the realified space
		\begin{align*}
			\cH_X^{\bR}:=\bR^d\otimes U_X\simeq \bR^{2d},
		\end{align*}
		where $U_X\simeq\bR^2$. On $U_X$ we use
		\begin{align*}
			\id_{U_X}:=\begin{pmatrix}1&0\\0&1\end{pmatrix}, \qquad \mathrm J_{U_X}:=\begin{pmatrix}0&-1\\1&0\end{pmatrix}.
		\end{align*}
		For a complex matrix $M$ acting on $\cH_X$ we define
		\begin{align*}
			\rea_X(M)\coloneqq\Re(M)\otimes\id_{U_X}+\Im(M)\otimes\mathrm J_{U_X}.
		\end{align*}
		The same formula is used for rectangular matrices $K:X_i\to X_o$, in which case $\rea_{X_oX_i}(K):X_i^{\bR}\to X_o^{\bR}$. By direct calculation one has
		\begin{align*}
			\rea_X(aM)&=a\rea_X(M),\qquad a\in\bR,\\
			\rea_X(M+N)&=\rea_X(M)+\rea_X(N),\\
			\rea_X(M^\dagger)&=\rea_X(M)^T,\\
			\rea_X(MN)&=\rea_X(M)\rea_X(N).
		\end{align*}
		Thus, if $\phi:L(\cH_{X_i})\to L(\cH_{X_o})$ is a complex CP map with Kraus form
		\begin{align*}
			\phi(\cdot)=\sum_r K_r\,\cdot\,K_r^\dagger,
		\end{align*}
		we define its realification by
		\begin{align*}
			\phi^{\bR}(\cdot)\coloneqq\sum_r \rea_{X_oX_i}(K_r)\,\cdot\,\rea_{X_oX_i}(K_r)^T.
		\end{align*}
		This is manifestly a real CP map. Moreover,
		\begin{align*}
			\sum_r \rea_{X_oX_i}(K_r)^T\rea_{X_oX_i}(K_r)=\rea_{X_i}\!\left(\sum_r K_r^\dagger K_r\right).
		\end{align*}
		Since $\rea$ is positive and unital on Hermitian matrices, trace non-increasing maps and trace-preserving maps are sent respectively to trace non-increasing and trace-preserving real maps.
		
			This gives a Kraus-level realification of CP maps. More precisely, for each complex CP map $\phi$ we choose a Kraus representation
		\begin{align*}
			\phi(\cdot)=\sum_r K_r\,\cdot\,K_r^\dagger
		\end{align*}
		and define one real CP map $\phi^{\bR}$ by realifying those chosen Kraus operators as above. We write
		\begin{align*}
			\enc_{X_i,X_o}(\cC(\phi))\coloneqq\cC(\phi^{\bR})
		\end{align*}
		for the Choi operator of this chosen realification. 
		By the calculation above, this chosen encoding sends CP maps to real CP maps and preserves trace non-increasing and trace-preserving maps. Therefore, choosing Kraus representations for all instrument elements, it sends complex instruments to valid real instruments.
		
		The second ingredient we need is the recomplexification procedure at the level of Kraus operators. To introduce it, we first select the following basis for the real matrices on $U_X$:
		\begin{align*}
			\id_{U_X}, \qquad \mathrm J_{U_X}, \qquad Z_{U_X}:=\begin{pmatrix}1&0\\0&-1\end{pmatrix}, \qquad \mathrm J_{U_X}Z_{U_X}.
		\end{align*}
		With this basis every real matrix $M$ on $\cH_X^{\bR}$ can be written as
		\begin{align*}
			M=M_1\otimes\id_{U_X}+M_2\otimes\mathrm J_{U_X}+M_3\otimes Z_{U_X}+M_4\otimes\mathrm J_{U_X}Z_{U_X}.
		\end{align*}
		Equivalently,
		\begin{align*}
			M=\rea_X(M_1+\iu M_2)+(\id_X\otimes Z_{U_X})\rea_X(M_3-\iu M_4).
		\end{align*}
		It is convenient to introduce the coisometry
		\begin{align*}
			V_X\coloneqq\frac{1}{\sqrt{2}}\left(\id_X\otimes\bra{0}_{U_X}+\iu\,\id_X\otimes\bra{1}_{U_X}\right),
		\end{align*}
		so that
		\begin{align*}
			M_1+\iu M_2&=V_XMV_X^\dagger,\\
			M_3-\iu M_4&=\bar V_XMV_X^\dagger.
		\end{align*}
		Therefore, if $K_r:X_i^{\bR}\to X_o^{\bR}$ is a real Kraus operator, we recomplexify it into the two complex Kraus operators
		\begin{align*}
			L_{r,0}&\coloneqq V_{X_o}K_rV_{X_i}^\dagger,\\
			L_{r,1}&\coloneqq \bar V_{X_o}K_rV_{X_i}^\dagger.
		\end{align*}
		So a real CP map
		\begin{align*}
			\psi(\cdot)=\sum_r K_r\,\cdot\,K_r^T
		\end{align*}
		is sent by $\rec$ to the complex CP map
		\begin{align*}
			\rec_{X_i,X_o}(\psi)(\cdot)\coloneqq\sum_r\left(L_{r,0}\,\cdot\,L_{r,0}^\dagger+L_{r,1}\,\cdot\,L_{r,1}^\dagger\right).
		\end{align*}

        This is CP by construction. Moreover, since
		\begin{align*}
		V_{X_o}^\dagger V_{X_o}+\bar V_{X_o}^\dagger\bar V_{X_o}=\id_{X_o^{\bR}},
		\end{align*}
		we have
		\begin{align*}
			\sum_r\left(L_{r,0}^\dagger L_{r,0}+L_{r,1}^\dagger L_{r,1}\right)
			&=
			V_{X_i}\left(\sum_r K_r^TK_r\right)V_{X_i}^\dagger .
		\end{align*}
		Hence, if $\psi$ is trace non-increasing, then
		\begin{align*}
		0\leq \sum_r K_r^TK_r\leq \id_{X_i^{\bR}},
		\end{align*}
		and therefore
		\begin{align*}
		0\leq 
		V_{X_i}\left(\sum_r K_r^TK_r\right)V_{X_i}^\dagger
		\leq
		V_{X_i}V_{X_i}^\dagger
		=
		\id_{X_i}.
		\end{align*}
		Thus $\rec_{X_i,X_o}(\psi)$ is trace non-increasing. If $\psi$ is trace preserving, then the same computation gives
		\begin{align*}
		\sum_r\left(L_{r,0}^\dagger L_{r,0}+L_{r,1}^\dagger L_{r,1}\right)
		=
		V_{X_i}\id_{X_i^{\bR}}V_{X_i}^\dagger
		=
		\id_{X_i},
		\end{align*}
		so $\rec_{X_i,X_o}(\psi)$ is trace preserving.
		
		At the Choi level, this Kraus-level recomplexification induces the decoding map
		\begin{align*}
			\dec_{X_i,X_o}(N)
			\coloneqq&
			(V_{X_i}\otimes\bar V_{X_o})N(V_{X_i}^\dagger\otimes\bar V_{X_o}^\dagger)\\
			&+
			(V_{X_i}\otimes V_{X_o})N(V_{X_i}^\dagger\otimes V_{X_o}^\dagger).
		\end{align*}
		Indeed, applying the above construction to a real Kraus decomposition of a real CP map gives precisely this expression for the Choi operator of the recomplexified map. Since every real positive Choi operator admits a real spectral decomposition, the preceding trace non-increasing and trace-preserving statements apply to all real CP maps.
		
		Finally, if $K=\rea_{X_oX_i}(\widetilde K)$, then
		\begin{align*}
			V_{X_o}KV_{X_i}^\dagger&=\widetilde K,\\
			\bar V_{X_o}KV_{X_i}^\dagger&=0.
		\end{align*}
		Thus recomplexification undoes the chosen realification at the Kraus level. In Choi form, for the chosen Kraus-level encodings used above, this gives
		\begin{align*}
			\dec_{X_i,X_o}(\enc_{X_i,X_o}(M))=M
		\end{align*}
		for every encoded Choi operator $M$ of a complex CP map.

		With these two ingredients we can finally prove the inclusion part of the theorem. We write it for two local agents, as in the statement. Let $W_{\QT}$ be a finite-dimensional bipartite QT process matrix, and let $\{M^A_{a|x,\QT}\}_a$ and $\{M^B_{b|y,\QT}\}_b$ be local QT instruments. We write
		\begin{align*}
			&M^{\QT}_{ab|xy}\coloneqq M^A_{a|x,\QT}\otimes M^B_{b|y,\QT}, \\
			& p(a,b|x,y)=\TR{W_{\QT}M^{\QT}_{ab|xy}}.
		\end{align*}
		Define
		\begin{align*}
			\dec_{AB}\coloneqq\dec_{A_i,A_o}\otimes\dec_{B_i,B_o}.
		\end{align*}
		For each instrument element, choose Kraus representations and set
		\begin{align*}
			M^{A,\RQT}_{a|x}&\coloneqq \enc_{A_i,A_o}(M^{A,\QT}_{a|x}),\\
			M^{B,\RQT}_{b|y}&\coloneqq \enc_{B_i,B_o}(M^{B,\QT}_{b|y}),\\
			M^{\RQT}_{ab|xy}&\coloneqq M^{A,\RQT}_{a|x}\otimes M^{B,\RQT}_{b|y}.
		\end{align*}
		By the properties of $\enc$, these are valid real CP maps, and for every choice of settings their sums are real CPTP maps. Hence they form valid RQT instruments.
		
		Now define
		\begin{align*}
			\widehat W\coloneqq\dec_{AB}^{\dagger}(W_{\QT}), \qquad W_{\RQT}\coloneqq\Re(\widehat W),
		\end{align*}
		where the adjoint is taken with respect to the Hilbert--Schmidt inner product. Since $\dec_{AB}$ is CP, also $\dec_{AB}^{\dagger}$ is CP. Hence $\widehat W\geq0$, because $W_{\QT}\geq0$. Moreover $W_{\RQT}$ is real and symmetric. For every real vector $v$,
		\begin{align*}
			v^TW_{\RQT}v=v^\dagger\widehat Wv\geq0.
		\end{align*}
		Thus $W_{\RQT}$ is positive semidefinite as a real symmetric matrix, and therefore also as a complex Hermitian operator.
		
		We now check that $W_{\RQT}$ is normalized as an RQT process matrix. Let $N^A$ and $N^B$ be arbitrary real CPTP Choi operators, and write
		\begin{align*}
			N^{AB}\coloneqq N^A\otimes N^B.
		\end{align*}
		Since $\dec_{A_i,A_o}(N^A)$ and $\dec_{B_i,B_o}(N^B)$ are complex CPTP Choi operators, we get
		\begin{align*}
			\TR{W_{\RQT}N^{AB}}&=\TR{\widehat WN^{AB}}=\TR{W_{\QT}\dec_{AB}(N^{AB})}=1.
		\end{align*}
		In the first equality we used that $N^{AB}$ is real symmetric, while the imaginary part of the Hermitian operator $\widehat W$ is real antisymmetric and therefore has zero trace against $N^{AB}$. Thus $W_{\RQT}$ is a valid bipartite RQT process matrix.
		
		Finally, the probabilities are preserved. For the chosen Kraus-level encodings above, recomplexification undoes realification on each instrument element. Using adjointness and $\dec_{X_i,X_o}(\enc_{X_i,X_o}(M))=M$ for those chosen encodings, we have
		\begin{align*}
			\TR{W_{\RQT}M^{\RQT}_{ab|xy}}&=\TR{\widehat WM^{\RQT}_{ab|xy}}=\TR{W_{\QT}\dec_{AB}(M^{\RQT}_{ab|xy})}=\TR{W_{\QT}M^{\QT}_{ab|xy}}=p(a,b|x,y).
		\end{align*}
		This proves that every finite-dimensional QT process correlation can be reproduced by a finite-dimensional RQT process experiment. Therefore
		\begin{align*}
			\mathcal C_{\mathrm{QT}}^{\proc,\fin}\subseteq\mathcal C_{\mathrm{RQT}}^{\proc,\fin}.
		\end{align*}
		
		\noindent\emph{Strictness.}
        In the End Matter, we provide a probability distribution realizable in RQT and make the LGYNI value strictly larger than QT.
        Below, to explain the see-saw algorithm we used to obtain an achievable value of $I_{\LGY}$ in RQT, we introduce notations.
        However, note that the following see-saw search is merely one method for finding the strict separation between QT and RQT. 
        Afterward, we will perform a final verification of the obtained solution.
        
        Let $I_{\LGY}(W,\bM^A,\bM^B)$ denote the LGYNI value calculated for the probability distribution realized by  \begin{equation}\label{eq:realization}
            p(a,b|x,y)=\Tr\left[W\left(M^{A_1A_2}_{a|x}\otimes M^{B_1B_2}_{b|y}\right)\right].
        \end{equation}
        Let $\mathscr M^A_{\mathrm{RQT}}$ denote the set of all pairs of real binary-outcome instruments for Alice's two settings.
    In other words, $\mathscr M^A_{\mathrm{RQT}}$ contains any set of unnormalized Choi operators $\bM^A=\{M^{A_1A_2}_{a|x}\}_{a,x}$ satisfying
    \begin{align}
    M^{A_1A_2}_{a|x} = \overline{M^{A_1A_2}_{a|x}}&=\left(M^{A_1A_2}_{a|x}\right)^T\quad(\forall a,\forall x),\\
    M^{A_1A_2}_{a|x}&\ge 0\quad(\forall a,\forall x),\\
    \Tr_{A_2}\left[\sum_a M^{A_1A_2}_{a|x}\right]&=\id^{A_1}\quad(\forall x).
    \end{align}
    The definition of $\mathscr M^B_{\mathrm{RQT}}$ is analogous.
    
    First, we define $\mathsf{OptW}(\bM^A,\bM^B)$ as 
    \begin{equation}\label{eq:optW}
        \mathsf{OptW}(\bM^A,\bM^B):=\underset{W\in\mathcal W_{\mathrm{RQT}}}{\arg\max}\,I_{\LGY}(W,\bM^A,\bM^B),
    \end{equation}
    that is, maximize the LGYNI value over the set $\mathcal W_{\RQT}$, that is, all RQT process matrices satisfying conditions 
    \begin{align}
        W=\overline{W}&=W^T,\\
        W&\ge 0,
    \end{align}
    and the normalization condition
    \begin{equation}\label{eq:W-RQT-normalization}
        \Tr[(M^{A_1A_2}\otimes M^{B_1B_2})W]=1
    \end{equation}
    for all unnormalized Choi operators of Alice and Bob's CPTP maps $M^{A_1A_2},M^{B_1B_2}\in\mathbb R^{4\times 4}$ satisfying
    \begin{align}
        M^{A_1A_2}\ge 0,\quad (M^{A_1A_2})^T=M^{A_1A_2},\quad\Tr_{A_2}M^{A_1A_2}&=\id^{A_1},\\
        M^{B_1B_2}\ge 0,\quad (M^{B_1B_2})^T=M^{B_1B_2},\quad\Tr_{B_2}M^{B_1B_2}&=\id^{B_1}.
    \end{align}
    This type of matrices can be decomposed into
    \begin{align}
        M^{A_1A_2}&=\frac{\id^{A_1}\otimes \id^{A_2}}{2}+F^{A_1A_2},\\
        M^{B_1B_2}&=\frac{\id^{B_1}\otimes \id^{B_2}}{2}+F^{B_1B_2},\\
    \end{align}
    such that $\Tr_{A_2}F^{A_1A_2}=\Tr_{B_2}F^{B_1B_2}=0$.
    Thus, the normalization condition Eq.~\eqref{eq:W-RQT-normalization} means that
    \begin{equation}
        \Tr\left[W\left(\frac{\id^{A_1}\otimes \id^{A_2}\otimes \id^{B_1}\otimes \id^{B_2}}{4}+\frac{\id^{A_1}\otimes \id^{A_2}\otimes F^{B_1B_2}}{2}+\frac{F^{A_1A_2}\otimes \id^{B_1}\otimes \id^{B_2}}{2}+F^{A_1A_2}\otimes F^{B_1B_2}\right)\right]=1
    \end{equation}
    must hold for all $F^{A_1A_2}$ and $F^{B_1B_2}$ such that $\Tr_{A_2}F^{A_1A_2}=\Tr_{B_2}F^{B_1B_2}=0$.
    Therefore, we check
    \begin{align}
        \Tr\left[W\left(\frac{\id^{A_1}\otimes \id^{A_2}\otimes \id^{B_1}\otimes \id^{B_2}}{4}\right)\right]&=1,\label{eq:normalization1}\\
        \Tr\left[W\left(\frac{\id^{A_1}\otimes \id^{A_2}\otimes F^{B_1B_2}}{2}\right)\right]&=0,\label{eq:normalization2}\\
        \Tr\left[W\left(\frac{F^{A_1A_2}\otimes \id^{B_1}\otimes \id^{B_2}}{2}\right)\right]&=0,\label{eq:normalization3}\\
        \Tr\left[W\left(F^{A_1A_2}\otimes F^{B_1B_2}\right)\right]&=0.\label{eq:normalization4}\\
    \end{align}
    for all $F^{A_1A_2}$ and $F^{B_1B_2}$ such that $\Tr_{A_2}F^{A_1A_2}=\Tr_{B_2}F^{B_1B_2}=0$.
    The condition Eq.~\eqref{eq:normalization1} implies that
    \begin{equation}
        \Tr W=4.
    \end{equation}
    Notice that real symmetric matrices are characterised by the following 10 basis vectors written by Pauli basis:
    \begin{align}
        &\{\id^{A_1}\otimes \id^{A_2},\;\id^{A_1}\otimes X^{A_2},\;\id^{A_1}\otimes Z^{A_2},\\
        &X^{A_1}\otimes \id^{A_2},\;X^{A_1}\otimes X^{A_2},\;X^{A_1}\otimes Z^{A_2},\\
        &Z^{A_1}\otimes \id^{A_2},\;Z^{A_1}\otimes X^{A_2},\;Z^{A_1}\otimes Z^{A_2},\\
        &Y^{A_1}\otimes Y^{A_2}\}.
    \end{align}
    However, $\Tr_{A_2}[\id^{A_1}\otimes \id^{A_2}],\;\Tr_{A_2}[X^{A_1}\otimes \id^{A_2}]$, and $\Tr_{A_2}[Z^{A_1}\otimes \id^{A_2}]$ are not zero, the only basis that generate $F^{A_1A_2}$ are the following 7 basis
    \begin{align}
        \{G^{A_1A_2}_i\}_{i=0}^{6}:=\{\id^{A_1}\otimes X^{A_2},\id^{A_1}\otimes Z^{A_2},X^{A_1}\otimes X^{A_2},\;X^{A_1}\otimes Z^{A_2},Z^{A_1}\otimes X^{A_2},\;Z^{A_1}\otimes Z^{A_2},Y^{A_1}\otimes Y^{A_2}\}.
    \end{align}
    Similarly, $F^{B_1B_2}$ is generated by 
    \begin{align}
        \{G^{B_1B_2}_j\}_{j=0}^{6}:=\{\id^{B_1}\otimes X^{B_2},\id^{B_1}\otimes Z^{B_2},X^{B_1}\otimes X^{B_2},\;X^{B_1}\otimes Z^{B_2},Z^{B_1}\otimes X^{B_2},\;Z^{B_1}\otimes Z^{B_2},Y^{B_1}\otimes Y^{B_2}\}.
    \end{align}
    Then, the conditions Eqs.~\eqref{eq:normalization2}, \eqref{eq:normalization3}, and \eqref{eq:normalization4} are
    \begin{align}
        r^B_j&:=\Tr[W(\id^{A_1}\otimes \id^{A_2}\otimes G^{B_1B_2}_j)]=0,\label{eq:residual_B}\\
        r^A_i&:=\Tr[W(G^{A_1A_2}_i\otimes \id^{B_1}\otimes \id^{B_2})]=0,\label{eq:residual_A}\\
        r^{AB}_{ij}&:=\Tr[W(G^{A_1A_2}_i\otimes G^{B_1B_2}_j)]=0\label{eq:residual_AB}
    \end{align}
    for all $i,j=0,\dots,6$.
    In this way, the normalization condition Eq.~\eqref{eq:W-RQT-normalization} for all Choi operators are rephrased to finite number constraints.

    Next, we define     $\mathsf{OptB}(W,\bM^A)$ as 
    \begin{equation}
        \mathsf{OptB}(W,\bM^A):=\underset{\bM^B\in\mathscr{M}^B_\RQT}{\arg\max}\,I_{\LGY}(W,\bM^A,\bM^B),
    \end{equation}
    which means maximizing the LGYNI value over all instruments in Bob's system when the process matrix $W$ and instrument $\bM^A$.
    Similarly, we define
    \begin{equation}
        \mathsf{OptA}(W,\bM^B):=\underset{\bM^A\in\mathscr{M}^A_\RQT}{\arg\max}\,I_{\LGY}(W,\bM^A,\bM^B).
    \end{equation}

    \begin{algorithm}[t]
        \caption{RQT see-saw search for LGYNI}
        \label{alg:rqt-lgyni-seesaw}
        \small
        \SetKwInOut{Input}{Input}
        \SetKwInOut{Output}{Output}
        \SetKwInOut{Parameter}{Parameter}
        
        \Input{RQT process matrix set $\mathcal W_{\RQT}$, and RQT instrument sets
        $\mathscr{M}^{A}_{\RQT}$ and $\mathscr{M}^{B}_{\RQT}$.}
        \Parameter{Random seed $1$, number of random starts $R=10$, sweeps $T_{\max}=20$, and inner alternating updates $K =5$.}
        \Output{$W^\star,\bM^{A,\star},\bM^{B,\star}$, and $I_{\LGY}^\star$.}
        
        \BlankLine
        Set \(I_{\LGY}^\star\leftarrow -\infty\).\;
        
        \For{\(r=1,\ldots,R\)}{
          Randomly initialize
          \(\bM^A\in\mathscr{M}^{A}_{\RQT}\) and
          \(\bM^B\in\mathscr{M}^{B}_{\RQT}\).\;
        
          Set
          \(W\leftarrow \mathsf{OptW}(\bM^A,\bM^B)\).\;
        
          Set
          \(I_{\LGY,\mathrm{best}}\leftarrow
          I_{\LGY}(W,\bM^A,\bM^B)\).\;
        
          Store $(W_{\mathrm{best}},\bM^A_{\mathrm{best}},\bM^B_{\mathrm{best}})
          \leftarrow
          (W,\bM^A,\bM^B)$.
        
          \For{\(t=1,\ldots,T_{\max}\)}{
            Set
            \(W\leftarrow \mathsf{OptW}(\bM^A,\bM^B)\).\;
        
            Set $\bM^A_{(0)}\leftarrow \bM^A,
            \quad
            \bM^B_{(0)}\leftarrow \bM^B$.
        
            \For{\(k=1,\ldots,K \)}{Set $\bM^B_{(k)}
              \leftarrow
              \mathsf{OptB}(W,\bM^A_{(k-1)})$.
              
              Set $\bM^A_{(k)}
              \leftarrow
              \mathsf{OptA}(W,\bM^B_{(k)})$.
            }
        
            Set $
            \bM^A\leftarrow \bM^A_{(K)},
            \quad
            \bM^B\leftarrow \bM^B_{(K)}$.
        
            Set
            \(W\leftarrow \mathsf{OptW}(\bM^A,\bM^B)\).\;
        
            Set
            \(I_{\LGY,\mathrm{new}}\leftarrow
            I_{\LGY}(W,\bM^A,\bM^B)\).\;
        
            \If{\(I_{\LGY,\mathrm{new}}>I_{\LGY,\mathrm{best}}\)}{
              Set \(I_{\LGY,\mathrm{best}}\leftarrow I_{\LGY,\mathrm{new}}\).\;
        
              Store $(W_{\mathrm{best}},\bM^A_{\mathrm{best}},\bM^B_{\mathrm{best}})\leftarrow(W,\bM^A,\bM^B)$.
            }
          }
        
          \If{$I_{\LGY,\mathrm{best}}>I_{\LGY}^\star$}{
            Set $I_{\LGY}^\star\leftarrow I_{\LGY,\mathrm{best}}$.\;
        
            Store $(W^\star,\bM^{A,\star},\bM^{B,\star})\leftarrow(W_{\mathrm{best}},\bM^A_{\mathrm{best}},\bM^B_{\mathrm{best}}).$
          }
        }
        
        \Return{\(W^\star,\bM^{A,\star},\bM^{B,\star},I_{\LGY}^\star\).}
    \end{algorithm}

    The process matrix and Choi operators obtained by running Algorithm~\ref{alg:rqt-lgyni-seesaw} are as follows.
    The process matrix is a $16\times 16$ matrix, so we show each $4\times 4$ block matrices
    \begin{equation}
        W^\star=
        \left(
        \begin{array}{cccc}
             W_{00}& W_{01} & W_{02} & W_{03} \\
             W_{10} & W_{11} & W_{12} & W_{13}\\
             W_{20}& W_{21} & W_{22} & W_{23} \\
             W_{30} & W_{31} & W_{32} & W_{33}
        \end{array}
        \right)
    \end{equation}
    Since $W$ is a symmetric matrix, $W_{ij}$ satisfies $W_{ji}=W_{ij}^T\quad (i<j)$.
    \begin{equation}
W_{11} =
\begin{pmatrix}
+1.754945120284\mathrm{e}{-01} & -7.340374654669\mathrm{e}{-03} & -8.596035581910\mathrm{e}{-02} & -1.339021238097\mathrm{e}{-01} \\
-7.340374654669\mathrm{e}{-03} & +1.750401343644\mathrm{e}{-01} & +5.975678640260\mathrm{e}{-02} & +9.294287726734\mathrm{e}{-02} \\
-8.596035581910\mathrm{e}{-02} & +5.975678640260\mathrm{e}{-02} & +4.966129674258\mathrm{e}{-01} & -1.190289682946\mathrm{e}{-01} \\
-1.339021238097\mathrm{e}{-01} & +9.294287726734\mathrm{e}{-02} & -1.190289682946\mathrm{e}{-01} & +3.970902156024\mathrm{e}{-01}
\end{pmatrix}.
\end{equation}

\begin{equation}
W_{12} =
\begin{pmatrix}
+1.229895586617\mathrm{e}{-01} & +4.788160920844\mathrm{e}{-06} & -2.937274161005\mathrm{e}{-03} & +3.325427933972\mathrm{e}{-04} \\
-4.788160917132\mathrm{e}{-06} & +1.229895586617\mathrm{e}{-01} & -1.561555938042\mathrm{e}{-04} & -3.388947510354\mathrm{e}{-03} \\
-3.385643033220\mathrm{e}{-03} & +1.561555938052\mathrm{e}{-04} & -1.229895586617\mathrm{e}{-01} & -6.616000573543\mathrm{e}{-06} \\
-3.325427933995\mathrm{e}{-04} & -2.933969683868\mathrm{e}{-03} & +6.616000576325\mathrm{e}{-06} & -1.229895586617\mathrm{e}{-01}
\end{pmatrix}.
\end{equation}

\begin{equation}
W_{13} =
\begin{pmatrix}
-9.251231135202\mathrm{e}{-02} & -8.753244624525\mathrm{e}{-04} & -3.037176353285\mathrm{e}{-02} & -4.812439868641\mathrm{e}{-02} \\
+9.677932360356\mathrm{e}{-04} & -9.250356090272\mathrm{e}{-02} & +2.224044705556\mathrm{e}{-02} & +3.510418465651\mathrm{e}{-02} \\
+3.626312032255\mathrm{e}{-02} & -2.301458436213\mathrm{e}{-02} & +9.137702761046\mathrm{e}{-02} & +1.694153139800\mathrm{e}{-03} \\
+4.986069991421\mathrm{e}{-02} & -3.152229816475\mathrm{e}{-02} & -1.492299295437\mathrm{e}{-04} & +9.202241061765\mathrm{e}{-02}
\end{pmatrix}.
\end{equation}

\begin{equation}
W_{14} =
\begin{pmatrix}
+6.773232720228\mathrm{e}{-02} & +6.778323880397\mathrm{e}{-04} & +2.280579870164\mathrm{e}{-02} & +3.608647206678\mathrm{e}{-02} \\
-6.778323880394\mathrm{e}{-04} & +6.773232720228\mathrm{e}{-02} & -1.666708760611\mathrm{e}{-02} & -2.628567706480\mathrm{e}{-02} \\
-2.628849143892\mathrm{e}{-02} & +1.666708760611\mathrm{e}{-02} & -6.773232720228\mathrm{e}{-02} & -6.764595992234\mathrm{e}{-04} \\
-3.608647206678\mathrm{e}{-02} & +2.280298432752\mathrm{e}{-02} & +6.764595992222\mathrm{e}{-04} & -6.773232720228\mathrm{e}{-02}
\end{pmatrix}.
\end{equation}

\begin{equation}
W_{21} =
\begin{pmatrix}
+1.229895586617\mathrm{e}{-01} & -4.788160917132\mathrm{e}{-06} & -3.385643033220\mathrm{e}{-03} & -3.325427933995\mathrm{e}{-04} \\
+4.788160920844\mathrm{e}{-06} & +1.229895586617\mathrm{e}{-01} & +1.561555938052\mathrm{e}{-04} & -2.933969683868\mathrm{e}{-03} \\
-2.937274161005\mathrm{e}{-03} & -1.561555938042\mathrm{e}{-04} & -1.229895586617\mathrm{e}{-01} & +6.616000576325\mathrm{e}{-06} \\
+3.325427933972\mathrm{e}{-04} & -3.388947510354\mathrm{e}{-03} & -6.616000573543\mathrm{e}{-06} & -1.229895586617\mathrm{e}{-01}
\end{pmatrix}.
\end{equation}

\begin{equation}
W_{22} =
\begin{pmatrix}
+2.531195799694\mathrm{e}{-01} & -7.340374654667\mathrm{e}{-03} & -8.795661797361\mathrm{e}{-02} & -1.339021238097\mathrm{e}{-01} \\
-7.340374654667\mathrm{e}{-03} & +2.526652023054\mathrm{e}{-01} & +5.975678640260\mathrm{e}{-02} & +9.094661511283\mathrm{e}{-02} \\
-8.795661797361\mathrm{e}{-02} & +5.975678640260\mathrm{e}{-02} & +4.189878994848\mathrm{e}{-01} & -1.190289682946\mathrm{e}{-01} \\
-1.339021238097\mathrm{e}{-01} & +9.094661511283\mathrm{e}{-02} & -1.190289682946\mathrm{e}{-01} & +3.194651476613\mathrm{e}{-01}
\end{pmatrix}.
\end{equation}

\begin{equation}
W_{23} =
\begin{pmatrix}
-1.258892827459\mathrm{e}{-01} & -1.248764958000\mathrm{e}{-03} & -4.199444252501\mathrm{e}{-02} & -6.650362785699\mathrm{e}{-02} \\
+1.248764958000\mathrm{e}{-03} & -1.258892827459\mathrm{e}{-01} & +3.072006114289\mathrm{e}{-02} & +4.846960646256\mathrm{e}{-02} \\
+4.846531664087\mathrm{e}{-02} & -3.072006114289\mathrm{e}{-02} & +1.258892827459\mathrm{e}{-01} & +1.249463486431\mathrm{e}{-03} \\
+6.650362785699\mathrm{e}{-02} & -4.199873234670\mathrm{e}{-02} & -1.249463486429\mathrm{e}{-03} & +1.258892827459\mathrm{e}{-01}
\end{pmatrix}.
\end{equation}

\begin{equation}
W_{24} =
\begin{pmatrix}
+9.215965603743\mathrm{e}{-02} & +9.670315454435\mathrm{e}{-04} & +3.151713992115\mathrm{e}{-02} & +4.985437886987\mathrm{e}{-02} \\
-8.745627718597\mathrm{e}{-04} & +9.216840648673\mathrm{e}{-02} & -2.301616551616\mathrm{e}{-02} & -3.627654693639\mathrm{e}{-02} \\
-3.512136084047\mathrm{e}{-02} & +2.224202820959\mathrm{e}{-02} & -9.329493977899\mathrm{e}{-02} & -1.482351761722\mathrm{e}{-04} \\
-4.811807764207\mathrm{e}{-02} & +3.036285571913\mathrm{e}{-02} & +1.693158386425\mathrm{e}{-03} & -9.264955677180\mathrm{e}{-02}
\end{pmatrix}.
\end{equation}

\begin{equation}
W_{31} =
\begin{pmatrix}
-9.251231135202\mathrm{e}{-02} & +9.677932360356\mathrm{e}{-04} & +3.626312032255\mathrm{e}{-02} & +4.986069991421\mathrm{e}{-02} \\
-8.753244624525\mathrm{e}{-04} & -9.250356090272\mathrm{e}{-02} & -2.301458436213\mathrm{e}{-02} & -3.152229816475\mathrm{e}{-02} \\
-3.037176353285\mathrm{e}{-02} & +2.224044705556\mathrm{e}{-02} & +9.137702761046\mathrm{e}{-02} & -1.492299295437\mathrm{e}{-04} \\
-4.812439868641\mathrm{e}{-02} & +3.510418465651\mathrm{e}{-02} & +1.694153139800\mathrm{e}{-03} & +9.202241061765\mathrm{e}{-02}
\end{pmatrix}.
\end{equation}

\begin{equation}
W_{32} =
\begin{pmatrix}
-1.258892827459\mathrm{e}{-01} & +1.248764958000\mathrm{e}{-03} & +4.846531664087\mathrm{e}{-02} & +6.650362785699\mathrm{e}{-02} \\
-1.248764958000\mathrm{e}{-03} & -1.258892827459\mathrm{e}{-01} & -3.072006114289\mathrm{e}{-02} & -4.199873234670\mathrm{e}{-02} \\
-4.199444252501\mathrm{e}{-02} & +3.072006114289\mathrm{e}{-02} & +1.258892827459\mathrm{e}{-01} & -1.249463486429\mathrm{e}{-03} \\
-6.650362785699\mathrm{e}{-02} & +4.846960646256\mathrm{e}{-02} & +1.249463486431\mathrm{e}{-03} & +1.258892827459\mathrm{e}{-01}
\end{pmatrix}.
\end{equation}

\begin{equation}
W_{33} =
\begin{pmatrix}
+1.621158448030\mathrm{e}{-01} & +7.340374654663\mathrm{e}{-03} & +9.013077421103\mathrm{e}{-02} & +1.339021238097\mathrm{e}{-01} \\
+7.340374654663\mathrm{e}{-03} & +1.625702224670\mathrm{e}{-01} & -5.975678640260\mathrm{e}{-02} & -8.877245887541\mathrm{e}{-02} \\
+9.013077421103\mathrm{e}{-02} & -5.975678640260\mathrm{e}{-02} & +1.657766757426\mathrm{e}{-01} & +1.190289682946\mathrm{e}{-01} \\
+1.339021238097\mathrm{e}{-01} & -8.877245887541\mathrm{e}{-02} & +1.190289682946\mathrm{e}{-01} & +2.652994275661\mathrm{e}{-01}
\end{pmatrix}.
\end{equation}

\begin{equation}
W_{34} =
\begin{pmatrix}
+6.689678968724\mathrm{e}{-03} & -4.163021640648\mathrm{e}{-06} & -4.030428170465\mathrm{e}{-04} & -3.357822690794\mathrm{e}{-04} \\
+4.163021642481\mathrm{e}{-06} & +6.689678968724\mathrm{e}{-03} & +1.536075913540\mathrm{e}{-04} & +4.560998468449\mathrm{e}{-05} \\
+4.889943233682\mathrm{e}{-05} & -1.536075913542\mathrm{e}{-04} & -6.689678968729\mathrm{e}{-03} & +6.803985468029\mathrm{e}{-06} \\
+3.357822690800\mathrm{e}{-04} & -3.997533693948\mathrm{e}{-04} & -6.803985464308\mathrm{e}{-06} & -6.689678968728\mathrm{e}{-03}
\end{pmatrix}.
\end{equation}

\begin{equation}
W_{41} =
\begin{pmatrix}
+6.773232720228\mathrm{e}{-02} & -6.778323880394\mathrm{e}{-04} & -2.628849143892\mathrm{e}{-02} & -3.608647206678\mathrm{e}{-02} \\
+6.778323880397\mathrm{e}{-04} & +6.773232720228\mathrm{e}{-02} & +1.666708760611\mathrm{e}{-02} & +2.280298432752\mathrm{e}{-02} \\
+2.280579870164\mathrm{e}{-02} & -1.666708760611\mathrm{e}{-02} & -6.773232720228\mathrm{e}{-02} & +6.764595992222\mathrm{e}{-04} \\
+3.608647206678\mathrm{e}{-02} & -2.628567706480\mathrm{e}{-02} & -6.764595992234\mathrm{e}{-04} & -6.773232720228\mathrm{e}{-02}
\end{pmatrix}.
\end{equation}

\begin{equation}
W_{42} =
\begin{pmatrix}
+9.215965603743\mathrm{e}{-02} & -8.745627718597\mathrm{e}{-04} & -3.512136084047\mathrm{e}{-02} & -4.811807764207\mathrm{e}{-02} \\
+9.670315454435\mathrm{e}{-04} & +9.216840648673\mathrm{e}{-02} & +2.224202820959\mathrm{e}{-02} & +3.036285571913\mathrm{e}{-02} \\
+3.151713992115\mathrm{e}{-02} & -2.301616551616\mathrm{e}{-02} & -9.329493977899\mathrm{e}{-02} & +1.693158386425\mathrm{e}{-03} \\
+4.985437886987\mathrm{e}{-02} & -3.627654693639\mathrm{e}{-02} & -1.482351761722\mathrm{e}{-04} & -9.264955677180\mathrm{e}{-02}
\end{pmatrix}.
\end{equation}

\begin{equation}
W_{43} =
\begin{pmatrix}
+6.689678968724\mathrm{e}{-03} & +4.163021642481\mathrm{e}{-06} & +4.889943233682\mathrm{e}{-05} & +3.357822690800\mathrm{e}{-04} \\
-4.163021640648\mathrm{e}{-06} & +6.689678968724\mathrm{e}{-03} & -1.536075913542\mathrm{e}{-04} & -3.997533693948\mathrm{e}{-04} \\
-4.030428170465\mathrm{e}{-04} & +1.536075913540\mathrm{e}{-04} & -6.689678968729\mathrm{e}{-03} & -6.803985464308\mathrm{e}{-06} \\
-3.357822690794\mathrm{e}{-04} & +4.560998468449\mathrm{e}{-05} & +6.803985468029\mathrm{e}{-06} & -6.689678968728\mathrm{e}{-03}
\end{pmatrix}.
\end{equation}

\begin{equation}
W_{44} =
\begin{pmatrix}
+1.649487144742\mathrm{e}{-01} & +7.340374654664\mathrm{e}{-03} & +9.005757699871\mathrm{e}{-02} & +1.339021238097\mathrm{e}{-01} \\
+7.340374654664\mathrm{e}{-03} & +1.654030921382\mathrm{e}{-01} & -5.975678640260\mathrm{e}{-02} & -8.884565608773\mathrm{e}{-02} \\
+9.005757699871\mathrm{e}{-02} & -5.975678640260\mathrm{e}{-02} & +1.629438060715\mathrm{e}{-01} & +1.190289682946\mathrm{e}{-01} \\
+1.339021238097\mathrm{e}{-01} & -8.884565608773\mathrm{e}{-02} & +1.190289682946\mathrm{e}{-01} & +2.624665578949\mathrm{e}{-01}
\end{pmatrix}.
\end{equation}

Now, we confirm that the obtained matrix $W^\star$ is a valid RQT process matrix.
Eigenvalues of the matrix $W^\star$ are as follows:
\begin{equation}
\{\lambda_k(W^\star)\}_{k=0}^{15}
=
\begin{aligned}[t]
(&
-1.643013259010\times 10^{-13}
 -7.490859234793\times 10^{-18} i,\\
&
-1.643013259010\times 10^{-13}
 +7.490859234793\times 10^{-18} i,\\
&
 2.451052110290\times 10^{-13},\,
 1.132603845210\times 10^{-12},\\
&
 3.794928772220\times 10^{-2},\,
 3.801383735231\times 10^{-2},\\
&
 7.814820297188\times 10^{-2},\,
 7.830699106535\times 10^{-2},\\
&
 1.774499050719\times 10^{-1},\,
 1.998970582305\times 10^{-1},\\
&
 4.619860995176\times 10^{-1},\,
 4.620507349149\times 10^{-1},\\
&
 4.999995979645\times 10^{-1},\,
 5.000004384173\times 10^{-1},\\
&
 6.660949049145\times 10^{-1},\,
 8.001029418553\times 10^{-1}
),
\end{aligned}
\end{equation}
where $i$ is the imaginary unit.
From the full expression of $W^\star$ given above, one finds that it is symmetric.
Thus, one finds that $W^\star$ is real and positive semidefinite with numerical error.

Next, we show Choi operators.
\begin{equation}
M^{A_1A_2}_{0|0} =
\begin{pmatrix}
+1.743455605170\mathrm{e}{-01} & +2.382780015555\mathrm{e}{-01} & -2.382780014899\mathrm{e}{-01} & +1.743455604302\mathrm{e}{-01} \\
+2.382780015555\mathrm{e}{-01} & +3.256544409647\mathrm{e}{-01} & -3.256544408352\mathrm{e}{-01} & +2.382780014899\mathrm{e}{-01} \\
-2.382780014899\mathrm{e}{-01} & -3.256544408352\mathrm{e}{-01} & +3.256544408358\mathrm{e}{-01} & -2.382780011393\mathrm{e}{-01} \\
+1.743455604302\mathrm{e}{-01} & +2.382780014899\mathrm{e}{-01} & -2.382780011393\mathrm{e}{-01} & +1.743455606462\mathrm{e}{-01}
\end{pmatrix}.
\end{equation}

\begin{equation}
M^{A_1A_2}_{1|0} =
\begin{pmatrix}
+1.743455594635\mathrm{e}{-01} & +2.382780001157\mathrm{e}{-01} & -2.382780000501\mathrm{e}{-01} & +1.743455593768\mathrm{e}{-01} \\
+2.382780001157\mathrm{e}{-01} & +3.256544389969\mathrm{e}{-01} & -3.256544388676\mathrm{e}{-01} & +2.382780000502\mathrm{e}{-01} \\
-2.382780000501\mathrm{e}{-01} & -3.256544388676\mathrm{e}{-01} & +3.256544388682\mathrm{e}{-01} & -2.382779996997\mathrm{e}{-01} \\
+1.743455593768\mathrm{e}{-01} & +2.382780000502\mathrm{e}{-01} & -2.382779996997\mathrm{e}{-01} & +1.743455595928\mathrm{e}{-01}
\end{pmatrix}.
\end{equation}

\begin{equation}
M^{A_1A_2}_{0|1} =
\begin{pmatrix}
+5.263673401007\mathrm{e}{-06} & -1.556727692761\mathrm{e}{-12} & +1.622270058133\mathrm{e}{-03} & -4.775607232277\mathrm{e}{-10} \\
-1.556727692761\mathrm{e}{-12} & +5.263672455923\mathrm{e}{-06} & -4.809796714743\mathrm{e}{-10} & +1.622269763597\mathrm{e}{-03} \\
+1.622270058133\mathrm{e}{-03} & -4.809796714743\mathrm{e}{-10} & +4.999947820747\mathrm{e}{-01} & -1.475528756081\mathrm{e}{-07} \\
-4.775607232277\mathrm{e}{-10} & +1.622269763597\mathrm{e}{-03} & -1.475528756081\mathrm{e}{-07} & +4.999946902934\mathrm{e}{-01}
\end{pmatrix}.
\end{equation}

\begin{equation}
M^{A_1A_2}_{1|1} =
\begin{pmatrix}
+6.481173262699\mathrm{e}{-01} & -4.775504513208\mathrm{e}{-01} & -2.105638157369\mathrm{e}{-03} & +1.545683140821\mathrm{e}{-03} \\
-4.775504513208\mathrm{e}{-01} & +3.518721463134\mathrm{e}{-01} & +1.551491841110\mathrm{e}{-03} & -1.138901664210\mathrm{e}{-03} \\
-2.105638157369\mathrm{e}{-03} & +1.551491841110\mathrm{e}{-03} & +6.841122191368\mathrm{e}{-06} & -5.021782049215\mathrm{e}{-06} \\
+1.545683140821\mathrm{e}{-03} & -1.138901664210\mathrm{e}{-03} & -5.021782049215\mathrm{e}{-06} & +3.686436135420\mathrm{e}{-06}
\end{pmatrix}.
\end{equation}

\begin{equation}
M^{B_1B_2}_{0|0} =
\begin{pmatrix}
+3.479047825674\mathrm{e}{-01} & -2.300320548007\mathrm{e}{-01} & +2.300320551413\mathrm{e}{-01} & +3.479047825878\mathrm{e}{-01} \\
-2.300320548007\mathrm{e}{-01} & +1.520954845253\mathrm{e}{-01} & -1.520954842846\mathrm{e}{-01} & -2.300320551414\mathrm{e}{-01} \\
+2.300320551413\mathrm{e}{-01} & -1.520954842846\mathrm{e}{-01} & +1.520954843624\mathrm{e}{-01} & +2.300320552024\mathrm{e}{-01} \\
+3.479047825878\mathrm{e}{-01} & -2.300320551414\mathrm{e}{-01} & +2.300320552024\mathrm{e}{-01} & +3.479047827225\mathrm{e}{-01}
\end{pmatrix}.
\end{equation}

\begin{equation}
M^{B_1B_2}_{1|0} =
\begin{pmatrix}
+3.479044108367\mathrm{e}{-01} & -2.300318090150\mathrm{e}{-01} & +2.300318093592\mathrm{e}{-01} & +3.479044108622\mathrm{e}{-01} \\
-2.300318090150\mathrm{e}{-01} & +1.520953220135\mathrm{e}{-01} & -1.520953217751\mathrm{e}{-01} & -2.300318093591\mathrm{e}{-01} \\
+2.300318093592\mathrm{e}{-01} & -1.520953217751\mathrm{e}{-01} & +1.520953218553\mathrm{e}{-01} & +2.300318094236\mathrm{e}{-01} \\
+3.479044108622\mathrm{e}{-01} & -2.300318093591\mathrm{e}{-01} & +2.300318094236\mathrm{e}{-01} & +3.479044110018\mathrm{e}{-01}
\end{pmatrix}.
\end{equation}

\begin{equation}
M^{B_1B_2}_{0|1} =
\begin{pmatrix}
+4.999174963751\mathrm{e}{-01} & -1.436302204857\mathrm{e}{-07} & -6.424405711519\mathrm{e}{-03} & +1.848233748633\mathrm{e}{-09} \\
-1.436302204857\mathrm{e}{-07} & +4.999173827334\mathrm{e}{-01} & +1.844813528150\mathrm{e}{-09} & -6.424404254207\mathrm{e}{-03} \\
-6.424405711519\mathrm{e}{-03} & +1.844813528150\mathrm{e}{-09} & +8.255969827816\mathrm{e}{-05} & -2.373907793215\mathrm{e}{-11} \\
+1.848233748633\mathrm{e}{-09} & -6.424404254207\mathrm{e}{-03} & -2.373907793215\mathrm{e}{-11} & +8.255967959004\mathrm{e}{-05}
\end{pmatrix}.
\end{equation}

\begin{equation}
M^{B_1B_2}_{1|1} =
\begin{pmatrix}
+1.125045753634\mathrm{e}{-04} & -7.693861980467\mathrm{e}{-05} & +8.771050644870\mathrm{e}{-03} & -5.962756110631\mathrm{e}{-03} \\
-7.693861980467\mathrm{e}{-05} & +5.261624256497\mathrm{e}{-05} & -5.998272019752\mathrm{e}{-03} & +4.077759319580\mathrm{e}{-03} \\
+8.771050644870\mathrm{e}{-03} & -5.998272019752\mathrm{e}{-03} & +6.838075290901\mathrm{e}{-01} & -4.648675953726\mathrm{e}{-01} \\
-5.962756110631\mathrm{e}{-03} & +4.077759319580\mathrm{e}{-03} & -4.648675953726\mathrm{e}{-01} & +3.160273514610\mathrm{e}{-01}
\end{pmatrix}.
\end{equation}

\begin{equation}
\lambda(M^{A_1A_2}_{0|0})
=
\left(
-4.068687645590\mathrm{e}{-11},
+3.910147083051\mathrm{e}{-11},
+4.749553155489\mathrm{e}{-10},
+1.000000002490\mathrm{e}{+00}
\right).
\end{equation}
Eigenvalues are as follows.
\begin{equation}
\lambda(M^{A_1A_2}_{1|0})
=
\left(
-4.068684900988\mathrm{e}{-11},
+3.910148288475\mathrm{e}{-11},
+4.749552121986\mathrm{e}{-10},
+9.999999964480\mathrm{e}{-01}
\right).
\end{equation}

\begin{equation}
\begin{split}
    \lambda(M^{A_1A_2}_{0|1})&=
    \left(
    +9.818705982117\mathrm{e}{-11}
    -1.248874415321\mathrm{e}{-19}i,
    +9.818705982117\mathrm{e}{-11}
    +1.248874415321\mathrm{e}{-19}i,\right.\\
    &\qquad\left.+4.999998452327\mathrm{e}{-01},
    +5.000001542849\mathrm{e}{-01}\right).
\end{split}
\end{equation}

\begin{equation}
\lambda(M^{A_1A_2}_{1|1})
=
\left(
+7.564617105936\mathrm{e}{-11},
+1.011410911642\mathrm{e}{-10},
+1.299110542561\mathrm{e}{-09},
+9.999999986657\mathrm{e}{-01}
\right).
\end{equation}

\begin{equation}
\lambda(M^{B_1B_2}_{0|0})
=
\left(
-4.069623133300\mathrm{e}{-11},
+3.910345746025\mathrm{e}{-11},
+4.750157206892\mathrm{e}{-10},
+1.000000533704\mathrm{e}{+00}
\right).
\end{equation}

\begin{equation}
\lambda(M^{B_1B_2}_{1|0})
=
\left(
-4.069626108065\mathrm{e}{-11},
+3.910348134046\mathrm{e}{-11},
+4.750156357491\mathrm{e}{-10},
+9.999994652340\mathrm{e}{-01}
\right).
\end{equation}

\begin{equation}
\lambda(M^{B_1B_2}_{0|1})
=
\left(
+9.783698329557\mathrm{e}{-11},
+9.783698346649\mathrm{e}{-11},
+4.999998446586\mathrm{e}{-01},
+5.000001536320\mathrm{e}{-01}
\right).
\end{equation}

\begin{equation}
\lambda(M^{B_1B_2}_{1|1})
=
\left(
+7.522642735038\mathrm{e}{-11},
+1.011585977520\mathrm{e}{-10},
+1.298743858098\mathrm{e}{-09},
+9.999999998940\mathrm{e}{-01}
\right).
\end{equation}

\begin{table}[tbp]
    \centering
    \begin{tabular}{ccccc}
    \toprule
    $(x,y)$
    & $p(0,0|x,y)$
    & $p(0,1|x,y)$
    & $p(1,0|x,y)$
    & $p(1,1|x,y)$\\
    \midrule
    $(0,0)$
    & $2.500001342868\mathrm{e}{-01}$
    & $2.499998671664\mathrm{e}{-01}$
    & $2.500001327763\mathrm{e}{-01}$
    & $2.499998656559\mathrm{e}{-01}$\\
    $(0,1)$
    & $4.436823888076\mathrm{e}{-01}$
    & $5.631761263742\mathrm{e}{-02}$
    & $4.436823861267\mathrm{e}{-01}$
    & $5.631761229709\mathrm{e}{-02}$\\
    $(1,0)$
    & $4.436826249996\mathrm{e}{-01}$
    & $4.436821509396\mathrm{e}{-01}$
    & $5.631764205055\mathrm{e}{-02}$
    & $5.631758187917\mathrm{e}{-02}$\\
    $(1,1)$
    & $1.637506815830\mathrm{e}{-01}$
    & $8.511686905092\mathrm{e}{-02}$
    & $8.503754563056\mathrm{e}{-02}$
    & $6.660949035921\mathrm{e}{-01}$\\
    \bottomrule
    \end{tabular}
    \caption{
    The distribution generated by the process matrix and Choi operators obtained by the see-saw algorithm.
    }
    \label{tab:rqt-lgyni-distribution-raw}
\end{table}

\begin{table}[t]
\centering
\begin{tabular}{lll}
\toprule
Object & Check & Value \\
\midrule
$W^\star$ (real)
& $\max_{ij}|\operatorname{Im} W^\star_{ij}|$
& $0$ \\
$W^\star$ (symmetric)
& $\|W^\star-W^{\star T}\|_{\max}$
& $0$ \\
$W^\star$ (positive semidefinite)
& $\lambda_{\min}(W^\star)$
& $-1.643013259010\mathrm{e}-13$ \\
$W^\star$ (normalization)
& $|\operatorname{Tr}W^\star-4|$
& $6.794564910706\mathrm{e}-13$ \\
$W^\star$ (normalization)
& $\max_i |r^A_i|$ (see Eq.\eqref{eq:residual_A})
& $2.342570581959\mathrm{e}-14$ \\
$W^\star$ (normalization)
& $\max_j|r^B_j|$ (see Eq.\eqref{eq:residual_B})
& $ 2.318978342686\mathrm{e}-14$  \\
$W^\star$ (normalization)
& $\max_{i,j}|r^{AB}_{ij}|$ (see Eq.\eqref{eq:residual_AB})
& $2.409836348213\mathrm{e}-14$ \\
\midrule
$\{M^{A_1A_2\star}_{a|x}\}_{a,x}$ (real)
& $\max_{a,x}\max_{ij}|\operatorname{Im}(M^{A_1A_2\star}_{a|x})_{ij}|$
& $0$ \\
$\{M^{A_1A_2\star}_{a|x}\}_{a,x}$ (symmetric)
& $\max_{a,x}\|M^{A_1A_2\star}_{a|x}-(M^{A_1A_2\star}_{a|x})^T\|_{\max}$
& $0$ \\
$\{M^{A_1A_2\star}_{a|x}\}_{a,x}$ (complete positivity)
& $\min_{a,x}\lambda_{\min}(M^{A_1A_2\star}_{a|x})$
& $-4.068687645590\mathrm{e}-11$ \\
$\{M^{A_1A_2\star}_{a|x}\}_{a,x}$ (trace-preserving)
& $\max_x\left\|\operatorname{Tr}_{A_2}\sum_a M^{A_1A_2\star}_{a|x}-\id_{A_1}\right\|_{\max}$ 
& $7.354583608787\mathrm{e}-11$ \\
$\{M^{A_1A_2\star}_{a|x}\}_{a,x}$ (trace non-increasing)
& $\min_{a,x}\lambda_{\min}\left(\id_{A_1}-\operatorname{Tr}_{A_2}M^{A_1A_2\star}_{a|x}\right)$
& $2.672758577671\mathrm{e}-10$ \\
\midrule
$\{M^{B_1B_2\star}_{b|y}\}_{b,y}$ (real)
& $\max_{b,y}\max_{ij}|\operatorname{Im}(M^{B_1B_2\star}_{b|y})_{ij}|$
& $0$ \\
$\{M^{B_1B_2\star}_{b|y}\}_{b,y}$ (symmetric)
& $\max_{b,y}\|M^{B_1B_2\star}_{b|y}-(M^{B_1B_2\star}_{b|y})^T\|_{\max}$
& $0$ \\
$\{M^{B_1B_2\star}_{b|y}\}_{b,y}$ (complete positivity)
& $\min_{b,y}\lambda_{\min}(M^{B_1B_2\star}_{b|y})$
& $-4.069626108065\mathrm{e}-11$ \\
$\{M^{B_1B_2\star}_{b|y}\}_{b,y}$ (trace-preserving)
& $\max_y\left\|\operatorname{Tr}_{B_2}\sum_b M^{B_1B_2\star}_{b|y}-\id_{B_1}\right\|_{\max}$
& $7.364897580686\mathrm{e}-11$ \\
$\{M^{B_1B_2\star}_{b|y}\}_{b,y}$ (trace non-increasing)
& $\min_{b,y}\lambda_{\min}\left(\id_{B_1}-\operatorname{Tr}_{B_2}M^{B_1B_2\star}_{b|y}\right)$
& $2.666775792365\mathrm{e}-10$ \\
\midrule
$\{p^\star(a,b|x,y)\}_{a,b,x,y}$
& $\max_{a,b,x,y}|\operatorname{Im}p^\star(a,b|x,y)|$
& $0$ \\
$\{p^\star(a,b|x,y)\}_{a,b,x,y}$
& $\min_{a,b,x,y}\operatorname{Re}p^\star(a,b|x,y)$
& $5.631758187917\mathrm{e}-02$ \\
$\{p^\star(a,b|x,y)\}_{a,b,x,y}$
& $\max_{x,y}\left|\sum_{a,b}p^\star(a,b|x,y)-1\right|$
& $1.433795304706\mathrm{e}-10$ \\
\bottomrule
\end{tabular}
\caption{
Numerical residuals for the RQT certificate. The norm
$\|\cdot\|_{\max}$ denotes the largest absolute entry. The rows
$R_{\mathrm A}$, $R_{\mathrm B}$, and $R_{\mathrm A\mathrm B}$ are the
largest residuals of the affine normalization constraints for the RQT
process matrix, evaluated on the real CPTP tangent bases. The small
negative eigenvalues are within the numerical solver tolerance.
}
\label{tab:rqt-certificate-residuals}
\end{table}

    Table~\ref{tab:rqt-lgyni-distribution-raw} shows the distribution generated by the process matrix and Choi operators.
    
    All SDPs used CVXPY~\cite{cvxpy} and MOSEK~\cite{mosek}.
    The MOSEK parameters were set to
    \begin{equation}
        \begin{aligned}
        \texttt{MSK\_DPAR\_INTPNT\_CO\_TOL\_PFEAS} &= 10^{-10},\\
        \texttt{MSK\_DPAR\_INTPNT\_CO\_TOL\_DFEAS} &= 10^{-10},\\
        \texttt{MSK\_DPAR\_INTPNT\_CO\_TOL\_REL\_GAP} &= 10^{-10},\\
        \texttt{MSK\_IPAR\_NUM\_THREADS}&=1.
        \end{aligned}
    \end{equation}
    The smallest eigenvalues were negative only at the level of the solver tolerance, and the imaginary parts of the raw eigenvalues were at the level of machine precision.

    The largest residual in Table~\ref{tab:rqt-certificate-residuals} is of order $10^{-10}$, whereas the separation gap between the reported RQT value and the complex-QT upper bound is of order $10^{-2}$.
    Thus the strict inequality is stable under the reported numerical residuals.

    The code is available at~\cite{code}.
    The version corresponding to the preprint is ``v0.1-preprint.''
	\end{proof}
    
\end{document}